\documentclass[a4paper,11pt,fleqn]{amsart}

\usepackage[utf8]{inputenc}   
\usepackage[T1]{fontenc}
\usepackage{lmodern}
\usepackage{microtype}
\usepackage{amssymb}
\usepackage{amsmath}
\usepackage{amsaddr} 
\usepackage{geometry}
\usepackage{enumerate}

\usepackage{tikz}
\usetikzlibrary{calc}

\usepackage[norelsize,ruled,linesnumbered]{algorithm2e}

\DeclareMathOperator{\aff}{aff}

\DeclareMathOperator{\cone}{cone}
\DeclareMathOperator{\conv}{conv}
\DeclareMathOperator{\dist}{dist}
\DeclareMathOperator{\size}{size}
\DeclareMathOperator{\vars}{Vars}

\newcommand{\citelist}[1]{\raisebox{.2ex}{[}#1\raisebox{.2ex}{]}}
\newcommand{\scite}[1]{\citeauthor{#1}, \citeyear{#1}}

\newcommand{\multiciteiii}[3]%
  {\citelist{\scite{#1}, \citeyear{#2}, \citeyear{#3}}}

\renewcommand{\phi}{\varphi}
\renewcommand{\epsilon}{\varepsilon}

\renewcommand{\emptyset}{\varnothing}  

\newtheorem{definition}{Definition}
\newtheorem{theorem}{Theorem}

\newtheorem{lemma}{Lemma}
\newtheorem{corollary}{Corollary}
\newtheorem{example}{Example}

\begin{document}

\title[Constraint Satisfaction and Semilinear Expansions of Addition]{Constraint Satisfaction and Semilinear Expansions of Addition over the Rationals and the Reals}
\date{}

\author{Peter Jonsson$^\ast$}
\thanks{$^\ast$The first author was partially supported by the Swedish Research Council (VR) under grant 621-2012-3239.}
\address{
Department of Computer and Information Science,
Link\"{o}pings universitet\\
SE-58183 Link\"{o}ping, Sweden}
\email{peter.jonsson@liu.se}
\author{Johan Thapper}
\address{
Laboratoire d'Informatique Gaspard-Monge,
Universit\'e Paris-Est Marne-la-Vall\'ee\\
5 boulevard Descartes, F-77420 Champs-sur-Marne,
France}
\email{thapper@u-pem.fr}

\pagestyle{plain}

\begin{abstract}
A semilinear relation is a 
finite union of finite intersections of open and closed half-spaces 
over, for instance, the reals, the rationals,
or the integers.
Semilinear
relations have been studied in connection with
algebraic geometry, automata theory, and
spatiotemporal reasoning. 
We consider semilinear relations over the rationals and the reals.
Under this assumption, the computational
complexity of the constraint satisfaction problem (CSP) is known
for all finite sets containing
$R_+=\{(x,y,z) \mid x+y=z\}$, $\leq$, and $\{1\}$.
These problems correspond to expansions of the linear programming feasibility problem.
We generalise this result and fully determine the complexity for all
finite sets of semilinear relations containing $R_+$.
This is accomplished in part by introducing an algorithm,
based on computing affine hulls,
which solves a new class of semilinear CSPs in polynomial time.
We further analyse the complexity of linear optimisation over
the solution set and the existence of integer solutions.

\end{abstract}

%

\maketitle

\section{Introduction}

We work over a ground set (or domain) $X$,
which will either be the rationals, ${\mathbb Q}$, or the reals, ${\mathbb R}$.
We say that a relation $R \subseteq X^k$ is {\em semilinear} if
it can be represented
as a finite union of finite intersections of open and closed half-spaces in $X^k$.
Alternatively, $R$ is semilinear if it is
first-order definable in $\{R_+,\leq,\{1\}\}$ where
$R_+=\{(x,y,z) \in X^3 \mid x+y=z\}$~\cite{Ferrante:Rackoff:sicomp75}. 
Semilinear relations appear in many different contexts within mathematics and
computer science: they are, for instance, frequently encountered in
algebraic geometry, automata theory, spatiotemporal reasoning, and computer algebra.
Semilinear relations have also attained a fair amount of attention in connection
with {\em constraint satisfaction problems} (CSPs). 
In a CSP, we are given a set of variables, a (finite or infinite) domain of values,
and a finite set of constraints.
The question is whether or not we can assign values to the variables
so that all constraints are satisfied.
From a complexity theoretical viewpoint,
solving general constraint satisfaction problems is obviously a hard problem.
Various ways of refining the problem can be adopted to allow a more meaningful analysis.
A common refinement is that of introducing a \emph{constraint language};
a finite set $\Gamma$ of allowed relations.
One then considers the problem CSP$(\Gamma)$ in which all constraint
in the input must be members of $\Gamma$.
This parameterisation of constraint
satisfaction problems has proved to be very fruitful for CSPs over both finite
and infinite domains. 
Since $\Gamma$ is finite,
the computational complexity of such a problem does not depend on the actual representation
of the constraints.

The complexity of finite-domain CSPs has been studied for a long time
and a powerful algebraic toolkit has gradually formed~\cite{Bulatov:etal:sicomp2005}.
Much of this work has been devoted to the 
{\em Feder-Vardi conjecture}~\cite{Feder:Vardi:stoc93,Feder:Vardi:siamjc98}
which posits that every finite-domain CSP
is either polynomial-time solvable or NP-complete.
Infinite-domain CSPs, on the other hand, constitute a much more diverse set of problems.
In fact, \emph{every} computational problem is polynomial-time equivalent to an infinite-domain
CSP~\cite{Bodirsky:Grohe:icalp2008}.
Obtaining a full understanding of their computational complexity is thus out of
the question,
and some further restriction is necessary.
In this article, this restriction will be to study semilinear relations and constraint languages.

A relation $R \subseteq X^k$ is said to be {\em essentially convex} if 
for all $p,q \in R$ there are only finitely many points on
the line segment between $p$ and $q$ that are not in $R$.
A constraint language $\Gamma$ is said to be essentially convex if every relation in $\Gamma$ is
essentially convex.
The main motivation for this study
is the following result:
\begin{theorem}[Bodirsky et al.~\cite{Bodirsky:etal:lmcs2012}] \label{bodirskytheorem}
Let $\Gamma$ be a finite set of semilinear relations over ${\mathbb Q}$ or ${\mathbb R}$ such that
$\{R_+,\leq,\{1\}\} \subseteq \Gamma$. Then,

\begin{enumerate}
\item
CSP$(\Gamma)$ is polynomial-time solvable if $\Gamma$ is essentially convex,
and
NP-complete otherwise; and
\item
the problem of optimizing a linear polynomial
over the solution set of CSP$(\Gamma)$ is polynomial-time solvable if and only if CSP$(\Gamma)$ is polynomial-time solvable (and NP-hard otherwise).
\end{enumerate}
\end{theorem}

One may suspect that there are semilinear constraint languages $\Gamma$ such that
CSP$(\Gamma) \in P$ but $\Gamma$ is not essentially convex. This is indeed
true and we identify two such cases.
In the first case, we consider relations with large ``cavities''.
It is not surprising that the algorithm for essentially convex relations
(and the ideas behind it) cannot be applied in the presence of
such highly non-convex relations.
Thus, we introduce a new algorithm which solves CSPs of this type in
polynomial time.
It is based on computing affine hulls
and the idea of improving an easily representable upper bound on the solution space by
looking at one constraint at a time; a form of ``local consistency'' method.
In the second case, we consider relations $R$ that are not necessarily essentially convex,
but look essentially convex when viewed form the origin.
That is, any points $p$ and $q$ that witnesses a not essentially convex relation
lie on a line that passes outside of the origin.
We show that we can remove all such
holes from $R$ to find an equivalent constraint language that is essentially convex,
and thereby solve the problem in polynomial time.

Combining these algorithmic results with matching
NP-hardness results and the fact that
CSP$(\Gamma)$ is always in NP for a semilinear constraint language $\Gamma$ 
(cf.~Theorem~5.2 in Bodirsky et al.~\cite{Bodirsky:etal:lmcs2012})
yields a dichotomy:

\begin{theorem} \label{superresult}
Let $\Gamma$ be a finite set of semilinear constraints that contains $R_+$.
Then, CSP$(\Gamma)$ is either in P or NP-complete.
\end{theorem}

Our result immediately generalises the first part of Theorem~\ref{bodirskytheorem}.
It also generalises another result by Bodirsky et al.~\cite{Bodirsky:etal:jlc2012}
concerning expansions of $\{R_+\}$ with relations that are first-order definable
in $\{R_+\}$.
One may note that this class of relations is a severely restricted subset
of the semilinear relations since it admits quantifier elimination
over the structure $\{+,\{0\}\}$, where $+$ denotes the binary addition function.
This follows from the more general fact that the
first-order theory of torsion-free divisible abelian groups admits quantifier elimination (see
e.g. Theorem 3.1.9 in \cite{Marker:MT}).
One may thus alternatively view relations that are first-order definable in $\{R_+\}$
as finite unions of sets defined by homogeneous linear systems of equations. 

We continue by generalising the second
part of Theorem~\ref{bodirskytheorem}, too: if $\Gamma$ is semilinear and contains $\{R_+,\{1\}\}$, then
the problem of optimising a linear polynomial
over the solution set of CSP$(\Gamma)$ is polynomial-time solvable if and only 
CSP$(\Gamma)$ is polynomial-time solvable (and NP-hard otherwise).
We also study the problem of finding integer solutions to
CSP$(\Gamma)$ for certain semilinear constraint languages $\Gamma$. 
Here, we obtain some partial results but a full classification remains elusive.
Our results shed some light on the {\em scalability property} introduced
by Jonsson and Lööw~\cite{Jonsson:Loow:AI13}.

This article has the following structure. We begin by formally defining constraint
satisfaction problems and semilinear relations together with some terminology and minor 
results in Section~\ref{sec:prels}.
The algorithms and tractability results that are presented in Section~\ref{sec:tractability}
while the hardness results can be found in Section~\ref{sec:hard}.
By combining the results from Section~\ref{sec:tractability} and \ref{sec:hard}, we prove
Theorem~\ref{superresult} in Section~\ref{sec:mainresult}.
We partially generalise Theorem~\ref{superresult}  to optimisation problems 
in Section~\ref{sec:optresult},
and we study the problem of finding integer solutions in Section~\ref{sec:intsol}.
Finally, we discuss some obstacles to further generalisations in Section~\ref{sec:waysforward}.
This article is a revised and extended version of a conference paper~\cite{Jonsson:Thapper:mfcs2014}.

\section{Preliminaries}
\label{sec:prels}

\subsection{Constraint satisfaction problems}

Let $\Gamma=\{R_1,\ldots,R_n\}$ be a finite set
of finitary relations over some domain $X$ (which will usually be infinite).
We refer to $\Gamma$ as
a \emph{constraint language}.
In order to avoid some uninteresting trivial cases, we will assume that all constraint languages are
non-empty and contain non-empty relations only.

A first-order formula is called \emph{primitive positive}
if it is of the form 
$\exists x_1,\dots,x_n. \psi_1 \wedge \dots \wedge \psi_m$,
where each $\psi_i$ is an atomic formula, i.e., either
$x=y$ or $R(x_{i_1},\dots,x_{i_k})$ with $R$ the relation
symbol for a $k$-ary relation from $\Gamma$. We call such
a formula a {\em pp-formula}. 
Note that all variables do not have to be existentially quantified;
if they are, then we say that the formula is a {\em sentence}.
Given a pp-formula $\Phi$, we let
$\vars(\Phi)$ denote the set of variables appearing in $\Phi$.
The atomic formulas $R(x_{i_1},\dots,x_{i_k})$ in a pp-formula $\Phi$ are also called the
\emph{constraints} of $\Phi$.

The \emph{constraint satisfaction problem for a constraint language $\Gamma$} 
(CSP$(\Gamma)$ for short)
is the following decision problem:

\medskip
\begin{center}
\fbox{
  \parbox{0.9\textwidth}{
{\bf Problem:} CSP$(\Gamma)$, where $\Gamma$ is a finite set of relations over a domain $X$.

\noindent
{\bf Input:} A primitive positive sentence $\Phi$ over $\Gamma$.

\noindent
{\bf Output:} `yes' if $\Phi$ is true in $\Gamma$, `no' otherwise.
}
}
\end{center}
\medskip

The exact representation of the relations in $\Gamma$ is unessential 
since we exclusively study finite constraint languages.

A relation $R(x_1,\ldots,x_k)$ is {\em pp-definable from}
$\Gamma$ if there exists a quantifier-free pp-formula $\phi$ 
over $\Gamma$ such that
\[R(x_1,\ldots,x_k) \equiv \exists y_1,\ldots,y_n.\phi(x_1,\ldots,x_k,y_1,\ldots,y_n).\]
The set of all relations that are pp-definable over $\Gamma$ is denoted by
$\langle \Gamma \rangle$.
The following easy but important result explains the role
of primitive positive definability for studying the computational complexity of CSPs.
We will use it extensively in the sequel without making explicit references.
\begin{lemma}[Jeavons~\cite{Jeavons:tcs98}] \label{pp-def}
Let $\Gamma$ be a constraint language and $\Gamma' \subseteq \langle \Gamma \rangle$ a
finite subset. 
Then CSP$(\Gamma')$ is polynomial-time
reducible to CSP$(\Gamma)$.
\end{lemma} 

Let $\Gamma=\{R_1,\ldots,R_k\}$ and $\Gamma'=\{R'_1,\ldots,R'_k\}$ be two
constraint languages such that $R_i$ and $R'_i$ are relations of the same arity.
Given an instance $\Phi$ of CSP$(\Gamma)$, let $\Phi'$ denote the instance
where each occurrence of a relation $R_i$ is replaced by $R'_i$.
We say that CSP$(\Gamma)$ is \emph{equivalent} to CSP$(\Gamma')$ if
$\Phi$ is true in $\Gamma$ if and only if $\Phi'$ is true in $\Gamma'$.
It is clear that if CSP$(\Gamma)$ and CSP$(\Gamma')$ are equivalent CSPs, then they
have the same complexity (up to a trivial linear-time transformation).

\subsection{Semilinear relations}

The domain, $X$, of every relation in this article will be 
the set of rationals, ${\mathbb Q}$,
or
the set of reals, ${\mathbb R}$.
In all cases, the set of coefficients, $Y$, will be the set of rationals,
but in order to avoid confusion, we will still make this explicit in our notation.
We define the following sets of relations.
\begin{itemize}
\item
$LE_X[Y]$ denotes the set of linear equalities over $X$ with coefficients in $Y$.
\item
$LI_X[Y]$ denotes the set of (strict and non-strict) linear inequalities over $X$ with coefficients in $Y$.
\end{itemize}

Sets defined by finite conjunctions of inequalities from $LI_X[Y]$ are called \emph{linear sets}.
The set of \emph{semilinear sets}, $SL_X[Y]$, is defined to be the set of finite unions
of linear sets.
We will refer to
$SL_{\mathbb Q}[{\mathbb Q}]$ and
$SL_{\mathbb R}[{\mathbb Q}]$ as semilinear relations over ${\mathbb Q}$ and ${\mathbb R}$,
respectively.
One should be aware of the representation of objects in $LE_X[Y]$ and $LI_X[Y]$
compared to $SL_X[Y]$. In $LE_X[Y]$ and $LI_X[Y]$, we view the equalities and inequalities
as syntactic objects which we can use for building logical formulas. Now, recall the
definition of a linear set: it is defined by a {\em conjunction} of inequalities.
However, a linear set is not a logical formula, it is a subset of $X^k$.
The same thing holds for semilinear sets: they are defined by {\em unions} of linear sets
and should thus not be viewed as logical formulas. This distinction has certain
advantages when it comes to terminology and notation but it also emphasise a
difference in the way we view and use these objects. 
The objects in $LE_X[Y]$ and $LI_X[Y]$ are often used in a logical context (such
as pp-definitions) while the semilinear relations 
are typically used in a geometric context.

Given a relation $R$ of arity $k$, let 
$R|_X = R \cap X^k$ and
$\Gamma|_X = \{ R|_X \mid R \in \Gamma \}$.
We demonstrate that CSP$(\Gamma)$ and CSP$(\Gamma|_{\mathbb Q})$ are equivalent
as constraint satisfaction problems whenever $\Gamma \subseteq SL_{\mathbb R}[{\mathbb Q}]$.
Thus, we will exclusively concentrate on relations from
$SL_{\mathbb Q}[{\mathbb Q}]$ in the sequel.
Let $\Gamma \subseteq SL_{\mathbb R}[{\mathbb Q}]$ and
let $\Phi$ be an instance of CSP$(\Gamma)$. Construct
an instance $\Phi'$ of CSP$(\Gamma|_{\mathbb Q})$ by replacing each occurrence
of $R$ in $\Phi$ by $R|_{\mathbb Q}$. 
If $\Phi'$ has a solution, then $\Phi$ has
a solution since $R|_{\mathbb Q} \subseteq R$ for each $R \in \Gamma$.
If $\Phi$ has a solution, then it has a rational solution by
Lemma 3.7 in Bodirsky et al.~\cite{Bodirsky:etal:lmcs2012}
so $\Phi'$ has a solution, too.

The following lemma is a direct consequence of our definitions: this particular 
property is often referred to as $o$-{\em minimality} in the literature~\cite{Hodges:ASMT}.
\begin{lemma}\label{fakeomin}
Let $R \in SL_X[Y]$ be a unary semilinear relation.
Then, $R$ can be written as a finite union of open, half-open, and closed intervals with endpoints in $Y \cup \{-\infty, \infty\}$ together
with a finite set of points in $Y$.
\end{lemma}

The set of semilinear relations can also be defined as those relations that are first-order definable
in $\{R_+, \leq, \{1\}\}$~\cite{Ferrante:Rackoff:sicomp75}.
In particular, $SL_X[Y]$ is closed under pp-definitions.
Consequently, Lemma~\ref{fakeomin} is applicable to all relations discussed in this article.

\begin{lemma}[Lemma~4.3 in Bodirsky et al.~\cite{Bodirsky:etal:jlc2012}] \label{generateequations}
Let $r_1,\ldots,r_k,r \in {\mathbb Q}$.
The relation 
$\{(x_1,\ldots,x_k) \in {\mathbb Q}^k \; | \; r_1x_1+\ldots+r_kx_k=r\}$
is pp-definable in $\{R_+,\{1\}\}$ and
it is pp-definable
in $\{R_+\}$ if $r = 0$.
Furthermore, the pp-formulas that define the relations can be computed in polynomial time.
\end{lemma}

It follows that 
$LE_{{\mathbb Q}}[{\mathbb Q}] \subseteq \langle \{R_+,\{1\}\} \rangle$ and
$LI_{{\mathbb Q}}[{\mathbb Q}] \subseteq \langle \{R_+, <, \leq, \{1\}\} \rangle$. 
One may also note that every homogeneous linear equation (with coefficients from ${\mathbb Q}$) 
is pp-definable in $\{R_+\}$.

\subsection{Unary semilinear relations}

For a rational $c$, and a unary relation $U \subseteq {\mathbb Q}$,
let $c \cdot U = \{ c \cdot x \mid x \in U \} \in \langle \{ R_+, U \} \rangle$.
When $c = -1$‚ we will also write $-U$ for $(-1) \cdot U$.

Given a relation $R \subseteq {\mathbb Q}^k$ and two distinct
points $a,b \in {\mathbb Q}^k$, we define
\[{\mathcal L}_{R,a,b}(y) \equiv \exists x_1,\ldots,x_k \: . \: R(x_1,\ldots,x_k) \wedge 
\textstyle\bigwedge_{i=1}^k x_i=(1-y) \cdot a_i+y \cdot b_i.\]
The relation ${\mathcal L}_{R,a,b}$ is a parameterisation of the intersection
between the relation $R$ and a line through the points $a$ and $b$.
Note that ${\mathcal L}_{R,a,b}$ is a member of $\langle LE_{\mathbb Q}[{\mathbb Q}] \cup \{R\} \rangle$ so, 
by Lemma~\ref{generateequations}, ${\mathcal L}_{R,a,b}$ is a member of
$\langle \{R_+,\{1\},R\} \rangle$, too.

A $k$-ary relation $R$ is {\em bounded} if
there exists an $a \in {\mathbb Q}$ such that $R \subseteq [-a,a]^k$.
A unary relation $U$ is \emph{unbounded in one direction} if 
$U$ is not bounded, but
there exists an $a \in {\mathbb Q}$ such that one of the following holds:
$U \subseteq [a,\infty)$; or $U \subseteq (-\infty,a]$.
A unary relation is called a \emph{bnu} (for \emph{bounded}, \emph{non-constant}, and \emph{unary}) if it is bounded and contains more than one point.

\begin{lemma} \label{halfboundgenerate}
Let $U$ be a unary relation in $SL_{\mathbb Q}[{\mathbb Q}]$
that is unbounded in one direction.
Then, 

\begin{enumerate}
\item
$\langle \{R_+, \{1\}, U \} \rangle$ contains
a bnu.
\item
if, in addition, $U$ contains both positive and negative elements, then
$\langle \{R_+,U\} \rangle$ contains a non-empty bounded unary relation.
\end{enumerate}
\end{lemma}

\begin{proof}
(1)
By Lemma~\ref{fakeomin}, there
exists an $a>0$ such that either

\begin{enumerate}[(i)]
\item\label{positivelyunbounded}
  $(-\infty,-a] \cap U = \emptyset$ and $[a,\infty) \subseteq U$; or
\item\label{negativelyunbounded}
  $(-\infty,-a] \subseteq U$ and $[a,\infty) \cap U = \emptyset$.
\end{enumerate}

Assume that (\ref{positivelyunbounded}) holds.
(The remaining case follows by considering $-U$.)
By choosing a rational $b > 2a$, it is not
hard to see
that the relation
\[
U'(x)\equiv\exists y \: . \: y=b-x \wedge U(x) \wedge U(y)
\]
is bounded and contains an interval.
The result then follows from Lemma~\ref{generateequations}.

\smallskip

(2) Assume that (\ref{positivelyunbounded}) holds and let $c \in U$ be a negative element.
(The remaining case follows by considering $-U$.)
Then,
\[
U''(x) \equiv \exists y \: . \: ay = cx \wedge U(x) \wedge U(y)
\]
is bounded and contains the element $a$.
The result again follows from Lemma~\ref{generateequations}.
\end{proof}

For a unary semilinear relation $T \subseteq {\mathbb Q}$, and a rational $\delta > 0$, let
$T + \mathcal{I}(\delta)$ denote the set of unary semilinear relations $U$ such that $T \subseteq U$
and
for all $x \in U$, there exists a $y \in T$ with $|x-y| < \delta$.

\begin{example}
The set $\{-1,1\} + \mathcal{I}(\frac{1}{2})$ contains all unary relations $U$ such that $\{-1,1\} \subseteq U \subseteq  (-\frac{3}{2}, -\frac{1}{2}) \cup (\frac{1}{2}, \frac{3}{2})$.
\end{example}

\begin{lemma}\label{lem:smalldelta1}
Let $U \neq \emptyset$ be a bounded unary semilinear relation such that $U \cap (-\infty, 0) = \emptyset$.
Then,
$\langle \{R_+, U\} \rangle$ contains a relation $U_\delta \in \{1\} + \mathcal{I}(\delta)$,
for every rational $\delta > 0$.
\end{lemma}

\begin{proof}
Let $U^+ = \sup U$ and $U^- = \inf U$.
By Lemma~\ref{fakeomin}, there exist elements $p^+, p^- \in U$ with 
$U^+ - p^+ < \delta U^+ $ and $p^- - U^- < \delta U^-$.
The relation $U_\delta := (p^-)^{-1} \cdot U \cap (p^+)^{-1} \cdot U$ is 
pp-definable in $\{R_+, U\}$ and satisfies:
$1 \in U_\delta$, $\sup U_\delta <  1 + \delta U^+(p^+)^{-1} \leq 1 + \delta$, and 
$\inf U_\delta > 1 - \delta U^-(p^-)^{-1} \geq 1-\delta$.
\end{proof}

\begin{lemma}\label{lem:smalldelta2}
Let $U$ be a bounded unary semilinear relation such that $U \cap (-\epsilon, \epsilon) = \emptyset$
for some $\epsilon > 0$ and $U \cap -U \neq \emptyset$.
Then,
$\langle \{R_+, U\} \rangle$ contains a relation $U_\delta \in \{-1, 1\} + \mathcal{I}(\delta)$,
for every rational $\delta > 0$.
\end{lemma}

\begin{proof}
Let $T = U \cap -U$. The proof then follows using a similar construction 
as in the proof of Lemma~\ref{lem:smalldelta1}.
\end{proof}

\subsection{Essential convexity}

Let $R$ be a $k$-ary relation over ${\mathbb Q}$.
The relation $R$ is {\em convex} if
for all $p,q \in R$, $R$ contains all points on the line
segment between $p$ and $q$.
We say that $R$ is \emph{essentially convex} if for all $p,q \in R$ there are only finitely many points on
the line segment between $p$ and $q$ that are not in $R$.

We say that $R$ {\em excludes an interval} if there are
$p,q \in R$ and real numbers $0 < \delta_1 < \delta_2 < 1$ such that
$p+(q-p)y \not\in R$ whenever $\delta_1 \leq y \leq \delta_2$.
Note that we can assume that
$\delta_1,\delta_2$ are rational numbers, since we can choose
any two distinct rational numbers $\gamma_1<\gamma_2$ between $\delta_1$ and $\delta_2$
instead of $\delta_1$ and $\delta_2$.

If $R$ is \emph{not} essentially convex, and if
 $p$ and $q$ are such that there are infinitely many points on the line
 segment between $p$ and $q$ that are not in $R$, 
 then we say that $p$ and $q$ {\em witness} that $R$ 
 is not essentially convex.
 Due to Lemma~\ref{fakeomin}, we conclude that a 
semilinear relation is essentially convex if and only if it does not exclude an interval.
We say that a constraint language is essentially convex if all its relations are essentially convex.

\begin{theorem}[Theorem~5.1 and~5.4 in Bodirsky et al.~\cite{Bodirsky:etal:lmcs2012}]\label{thm:essconvtractable}
If $\Gamma$
is a finite set of
essentially convex
semilinear relations,
then CSP$(\Gamma)$ is in P.
\end{theorem}

\section{Tractability}\label{sec:tractability}

In this section, we present our two main sources of tractability.
Section~\ref{sec:affine} contains a new algorithm for semilinear
constraint languages $\Gamma$ containing $\{R_+, \{1\}\}$ and such
that $\langle \Gamma \rangle$ does not contain a bnu.
In Section~\ref{sec:moreess}, we extend the applicability of
Theorem~\ref{thm:essconvtractable} from essentially convex semilinear constraint
languages to a certain class of semilinear CSPs that are not essentially convex.

\subsection{Affine consistency}\label{sec:affine}

Instead of computing the exact solution set to a CSP instance, our approach
will be to reduce an upper bound on this set as far as possible.
In particular, we will maintain a representation of an \emph{affine subspace} that
is guaranteed to contain the solution set, and repeatedly intersect this
subspace with every constraint in order to attempt to reduce it further.
This can be seen as a form of \emph{local consistency}.
If we manage to reduce the upper bound to an empty set, then we are
certain that the instance is unsatisfiable.
We will show that under certain conditions, the converse holds;
if the upper bound is non-empty, then there are necessarily solutions.
To formalise this idea, we will need some definitions.

For a subset $S \subseteq {\mathbb Q}^n$,
let $\aff(S)$ denote the \emph{affine hull of $S$ in ${\mathbb Q}^n$}:
\[\aff(S) = \{ \sum_{i=1}^k x_i p_i \mid k \geq 1, x_i \in {\mathbb Q},
p_i \in S, \sum_{i=1}^k x_i = 1 \}.\]
An \emph{affine subspace} is a subset $S \subseteq {\mathbb Q}^n$ for
which $\aff(S) = S$.
The points $p_1, \dots, p_k \in {\mathbb Q}^n$ are said to be
\emph{affinely independent} if $x_1 p_1 + \dots + x_k p_k = 0$ with $x_1
+ \dots + x_k = 0$ implies $x_1 = \dots = x_k = 0$.
The dimension, $\dim(S)$, of a set $S \subseteq {\mathbb Q}^n$ is defined to be
one less than the maximum number of affinely independent points in $S$.

We define a notion of consistency for sets of semilinear constraints which we call
{\em affine consistency}.
Let $V$ be a finite set of variables and let $n = |V|$.
A set of constraints $R_i(x_{i_1}, \dots, x_{i_k})$ with $\{x_{i_1},\dots,x_{i_k}\} \subseteq V$ is \emph{affinely consistent with respect
to a non-empty affine subspace} $\emptyset \neq A \subseteq {\mathbb Q}^V$ if
$\aff(\hat{R}_i \cap A) = A$ for all $i$,
where 
$\hat{R}_i := \{ (x_1, \dots, x_n) \in {\mathbb Q}^V \mid (x_{i_1}, \dots, x_{i_k}) \in R_i \}$.

\begin{algorithm}[ht]
\SetAlgoLined
\DontPrintSemicolon
\KwIn{A set of constraints $\{ R_i(x_{i_1}, \dots, x_{i_k}) \}$ over variables $V$}
\KwOut{``yes'' if the resulting affine subspace is non-empty, ``no'' otherwise}

$A := {\mathbb Q}^V$\\
\Repeat{$A$ {\rm does not change}}{
\ForEach{{\rm constraint} $R_i(x_{i_1}, \dots, x_{i_k})$}{
$A := \aff(\hat{R}_i \cap A)$\\
}
}
\leIf{$A \neq \emptyset$}{
   \Return ``yes''
}{
   \Return ``no''
}

\caption{Affine consistency}
\label{alg:affc}
\end{algorithm}

Algorithm~\ref{alg:affc} establishes affine consistency for a set of constraints and
answers ``yes'' if the resulting affine subspace is non-empty and ``no'' otherwise.
In the rest of this section, we show that this algorithm correctly solves CSP$(\Gamma)$
when $\{R_+, \{1\}\} \subseteq \Gamma$ is a semilinear constraint language such
that $\langle \Gamma \rangle$ does not contain a bnu.
Furthermore, we show that the algorithm can be implemented
to run in polynomial time when applied to constraint languages of this kind.

We begin by proving a technical lemma which is the basis for these results.

\begin{lemma}\label{lem:intersection}
Let $P = P_1 \cup \dots \cup P_k, Q = Q_1 \cup \dots \cup Q_l \in
SL_{\mathbb Q}[{\mathbb Q}]$
be two $n$-ary relations and $P_1,\dots,P_k,Q_1,\dots,Q_l$ linear sets. 
Assume that neither
$\langle LE_{\mathbb Q}[{\mathbb Q}] \cup \{ P \} \rangle$ nor $\langle
LE_{\mathbb Q}[{\mathbb Q}] \cup \{ Q \} \rangle$ contains a bnu.
If $\aff(P) = \aff(Q) =: A$, then $\aff(P_i \cap Q_j) = A$ for some $i$
and $j$.
\end{lemma}

\begin{proof}
The proof is by induction on the dimension $d = \dim(A)$.
For $d = 0$, both $P$ and $Q$ consist of a single point $p$.
Clearly, $P_i = \{p\}$ for some $i$ and $Q_j = \{p\}$ for some $j$.
Now assume that $d > 0$ and that the lemma holds for all $P'$, $Q'$ with
$\aff(P') = \aff(Q') = A'$ and $\dim(A') < d$.
Let $p_0, p_1, \dots, p_d$ be $d+1$ affinely independent points in $P$
and let $q_0, q_1, \dots, q_d$ be $d+1$ affinely independent points in $Q$.
For $1 \leq i \leq d$,
consider the lines $L^p_i$ through $p_0$ and $p_i$, and the lines
$L^q_i$ through $q_0$ and $q_i$.
Let $H = \{ y  \in {\mathbb Q}^n \mid \alpha \cdot y = 0 \}$ ($\alpha
\in {\mathbb Q}^n$) be a hyperplane  in ${\mathbb Q}^n$ through the
origin that is not parallel to any of the lines $L^p_i$ or $L^q_i$.
Then, $H$ intersects each of the $2d$ lines.
Let $H(c) = \{ y \in {\mathbb Q}^n \mid \alpha \cdot y = c \}$ and
let $B(c) = \{ y \in {\mathbb Q}^n \mid \alpha \cdot y \not\in [-c,c] \}$.

Express the line $L^p_i$ as $\{ y \in {\mathbb Q}^n \mid y = (1-x) \cdot a +
x \cdot b, x \in {\mathbb Q} \}$, for some $a, b \in {\mathbb Q}^n$.
Let $T = \mathcal{L}_{P,a,b} \in \langle LE_{\mathbb Q}[{\mathbb Q}] \cup \{ P \} \rangle$.
Since $T$ contains $p_0$ and $p_i$, it follows that $T$ is not a
constant and hence unbounded.
By Lemma~\ref{halfboundgenerate}(1), $T$ is unbounded in both directions.
By Lemma~\ref{fakeomin},
$B(c^p_i) \cap L^p_i \subseteq T \subseteq P$, for some positive
constant $c^p_i$.
An analogous argument shows that that $B(c^q_j) \cap L^q_j \subseteq Q$,
for some positive constant $c^q_j$.
Let $c'$ be a positive constant such that $p_0, q_0 \not\in B(c')$ and
let $c = \max \{c'\} \cup \{ c^p_i, c^q_j \mid 1 \leq i,j \leq d \}$.
This ensures that for any $x > c$, $H(x) \cap P$ intersects the lines
$L^p_i$ in $d$ affinely independent points, and $H(x) \cap Q$
intersects the lines $L^q_j$ in $d$ affinely independent points.

Let $P' = H(x) \cap P$, $P'_i = H(x) \cap P_i$, $Q' = H(x) \cap Q$, and
$Q'_j = H(x) \cap Q_j$.
We now have $\aff(P') = \aff(Q') = A'(x)$ with $\dim(A'(x)) = \dim(A)-1$, for every $x > c$.
By induction on $P' = P'_1 \cup \dots \cup P'_k$ and $Q' = Q'_1 \cup \dots \cup Q'_l$,
it follows that 
$\aff(H(x) \cap (P_{i(x)} \cap Q_{j(x)})) = \aff(P'_{i(x)} \cap Q'_{j(x)}) = A'(x)$
for some $i(x)$ and $j(x)$. This holds for all $x > c$, hence there
exist distinct $x_1, x_2 > c$ with $i(x_1) = i(x_2) = i'$ and $j(x_1) =
j(x_2) = j'$.
Since $A'(x_1), A'(x_2) \subseteq \aff(P_{i'} \cap Q_{j'})$, $A'(x_1)
\cap A'(x_2) = \emptyset$, and $\dim(A'(x_2)) = d-1 \geq 0$, it follows
that $\aff(P_{i'} \cap Q_{j'})$ strictly contains $A'(x_1)$, so we have
$A' \subset \aff(P_{i'} \cap Q_{j'}) \subseteq A$, and $\dim(A'(x_1)) =
\dim(A)-1$.
Therefore we have the equality $\aff(P_{i'} \cap Q_{j'}) = A$.
The lemma follows.
\end{proof}

\begin{algorithm}[ht]
\SetAlgoLined
\DontPrintSemicolon
\KwIn{A semilinear relation $R = R_1 \cup \dots \cup R_k$ and an affine subspace $A$.}
\KwOut{A set of inequalities defining $\aff(R \cap A)$, or $\bot$ if $\aff(R \cap A) = \emptyset$.}

Find $i$ that maximises $d_i := \dim(\aff(R_i \cap A)).$\\
\lIf{$\aff(R_i \cap A) = \emptyset$}{\Return $\bot$}
Let $I$ be the set of inequalities for $R_i$ and $J$ be the set of inequalities for $A$.\\
$S := I \cup J$\\
\ForEach{{\rm inequality} $\iota \in I \cup J$}{
\If{$\dim(\aff(S \setminus \{ \iota \})) = d_i$}{
      $S := S \setminus \{ \iota \}$
}
}
\Return $S$

\caption{Calculate $\aff(R \cap A)$}
\label{alg:aff}
\end{algorithm}

For a semilinear relation $R$, we let $\size(R)$ denote the \emph{representation size}
of $R$, i.e., the number of bits needed to describe the arities and coefficients of
each inequality in some fixed definition of $R$.

\begin{lemma}
Let $R \in SL_{\mathbb Q}[{\mathbb Q}]$ be a relation such that $\langle
LE_{\mathbb Q}[{\mathbb Q}] \cup \{ R \} \rangle$ does not contain a bnu and
let $A \subseteq {\mathbb Q}^n$ be an affine subspace.
Algorithm \ref{alg:aff} computes a set of linear inequalities $S$
defining $\aff(R \cap A)$ in time polynomial in $\size(R)+\size(A)$ and with
$\size(S) \leq \size(R)+\size(A)$.
\end{lemma}

\begin{proof}
Let $R = R_1 \cup \dots \cup R_k$ be the representation of $R$ as the union of
linear sets $R_i$.
By Lemma~\ref{lem:intersection}, there exists an $i$ such that $\aff(R
\cap A) = \aff(R_i \cap A)$ and since $\aff(R_j \cap A) \subseteq \aff(R
\cap A)$ for all $j$, the algorithm will find such an $i$ on line 1
by simply comparing the
dimensions of these sets.
If $\aff(R \cap A) = \emptyset$, then the algorithm returns $\bot$, signalling that
the affine hull is empty.

Otherwise, the affine hull of a non-empty polyhedron can always be obtained
as a subset of its defining inequalities (cf.~Schrijver~\cite[Section~8.2]{Schrijver:TLIP}).
Here, some of the inequalities may be strict, but it is not hard to see that removing them
does not change the affine hull.
If $\iota \in I \cup J$ is an inequality that cannot be removed without
increasing the dimension of the affine hull, then it is clear that $\iota$
still cannot be removed after the loop.
Hence, after the loop, no inequality in $S$ can be removed without
increasing the dimension of the affine hull.
It follows that $S$ itself defines an affine subspace, $A_S$,
and $A_S = \aff(A_S) = \aff(R_i \cap A) = \aff(R \cap A)$.

Using the ellipsoid method, we can determine the dimension of the affine
hull of a polyhedron defined by a system of linear inequalities in time
polynomial in the representation size of the
inequalities~\cite[Corollary 14.1f]{Schrijver:TLIP}.
To handle strict inequalities on line 1,
we can perturb these by a small amount, while keeping the
representation sizes polynomial, to obtain a system of non-strict
inequalities with the same affine hull.
The algorithm does at most $|I \cup J|+k$ affine hull calculations.
The total time is thus polynomial in $\size(R)+\size(A)$.
Finally, the set $S$ is a subset of $I \cup J$,
so $\size(A_S) \leq \size(R)+\size(A)$.
\end{proof}

\begin{theorem}\label{thm:affinesolves}
Let $\{R_+, \{1\}\} \subseteq \Gamma \subseteq SL_{\mathbb
Q}[{\mathbb Q}]$ be a finite constraint language.
If there is no bnu in $\langle \Gamma \rangle$, then Algorithm \ref{alg:affc}
correctly
solves CSP$(\Gamma)$ and can be implemented to run in polynomial time.
\end{theorem}

\begin{proof}
Assume that each relation $R \in \Gamma$ is given as $R = R_1 \cup \dots \cup
R_k$, where $R_i$ is a linear set for each $i$.
First, we show that the algorithm terminates with $A$ equal to the
affine hull of the solution space of the constraints.

Assume that the input consists of the constraints $R_i(x_{i_1},\dots,x_{i_k})$
over variables $V$, $i = 1, \dots, m$. Let $Z = \bigcap_{i=1}^m \hat{R}_i$
denote the solution space of the instance.
It is clear that $Z$ is contained in $A$ throughout the execution of the
algorithm.
Therefore, $\aff(Z) = \aff(Z \cap A)$ so it suffices to show that $\aff(Z \cap A) = A$.
We will show that $\aff(\bigcap_{i=1}^j \hat{R}_i \cap A)  = A$ for all $j = 1, \dots, m$.
When the algorithm terminates, we have $\aff(\hat{R}_i \cap A) = A$
for every $i = 1, \dots, m$.
In particular, the claim holds for $j = 1$.
Now assume that the claim holds for $j-1$.
Then, $P = \bigcap_{i=1}^{j-1} \hat{R}_i \cap A$ and $Q = \hat{R}_j \cap A$ satisfy
the requirements of Lemma~\ref{lem:intersection} with $\aff(P) = \aff(Q) = A$.
Therefore, we can use
this lemma to conclude that $\aff(\bigcap_{i=1}^{j} \hat{R}_i \cap A) = \aff(P \cap Q) = A$.

Finally, we show that the algorithm can be implemented to run in
polynomial time.
The call to Algorithm~\ref{alg:aff} in the inner loop is carried out at
most $mn$ times,
where $n = |V|$.
The size of $\hat{R}$ is at most $\size(R) + \log n$, so the size of $A$
never exceeds
$\mathcal{O}(mn (\size(R) + \log n) )$, where $R$ is a
relation with maximal representation size.
Therefore, each call to Algorithm~\ref{alg:aff} takes polynomial time
and consequently, the entire algorithm runs in polynomial time.
\end{proof}

\subsection{Essential convexity}\label{sec:moreess}

We will now identify another family of polynomial-time solvable semilinear CSPs.
This time, we base our result on essentially convex semilinear
constraint languages (Theorem~\ref{thm:essconvtractable}).
We extend this result to the situation where we are only
guaranteed that all unary relations that are pp-definable in the language are
essentially convex.
The idea is that even if we do not have the constant relation
$\{1\}$ to help us identify excluded intervals, we are still able to see excluded
full-dimensional holes.
We follow up this by showing that we can remove certain lower-dimensional holes
and thus recover an equivalent essentially convex constraint language.
We remind the reader that the dimension of a set is defined with respect to its affine hull,
as in Section~\ref{sec:affine}.

For $x, y \in {\mathbb Q}^k$, we let $\|x\|$ denote the euclidean norm of $x$,
and $\dist(x,y) = \|x-y\|$ the euclidean distance between $x$ and $y$.

\begin{lemma}\label{lemur}
Let $U \in \{1\} + \mathcal{I}(c)$ for some $0 < c < 1$ and
assume that $R \in SL_{\mathbb Q}[{\mathbb Q}]$ is a semilinear relation such that every unary relation in
$\langle \{R_+, U, R \} \rangle$ is essentially convex.
Then, $R$ can be defined by a formula
$\phi_0 \wedge \neg \phi_1 \wedge \dots \wedge \neg \phi_k$, where $\phi_0$ defines a
convex semilinear set, and $\phi_1, \dots, \phi_k$
are conjunctions over $LI_{\mathbb Q}[{\mathbb Q}]$
that define convex sets
of dimensions strictly lower than the dimension of the set defined by $\phi_0$.
\end{lemma}

\begin{proof}
Let $\conv(R)$ denote the convex hull of $R$ and let $d$ denote its dimension.
The set $\conv(R)$ is semilinear (see, for instance, Stengle et al.~\cite{Stengle2010370}).
Let $\phi_0$ be a formula for $\conv(R)$ and
let $\phi_1 \vee \dots \vee \phi_k$ be a formula for $\conv(R) \setminus R$ on quantifier-free DNF over $LI_{\mathbb Q}[{\mathbb Q}]$.
It remains to show that for each $i$, the dimension of the convex set $S_i$ defined by $\phi_i$ is smaller than $d$.
To prove this, we show that for every point $p$ in $S_i$,
and every $\epsilon > 0$,
there exists a point
$x$ in $R$ such that $\dist(p, x) < \epsilon$.
Since every $d$-dimensional convex set contains a small $d$-dimensional open ball around
every point in its interior, it follows from this that none of the sets $S_i$ can be $d$-dimensional.

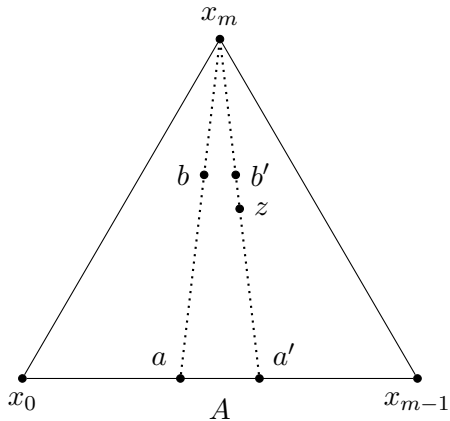
\begin{figure}[ht]\label{fig:simplex}
\begin{center}
\begin{tikzpicture}
[
scale = 3.0,
point/.style = {draw, circle, fill=black, inner sep=1pt}
]

\node (n0) at +(90:1) [point, label = above:$x_m$] {};
\node (n1) at +(-30:1) [point, label = {below:$x_{m-1}$}] {};
\node (n2) at +(-150:1) [point, label = {below:$x_0$}] {};

\draw (n0) -- (n1) -- (n2) -- (n0);

\node (a) at ($(n1)!.6!(n2)$) [point, label = {above left:$a$}] {};
\node (a2) at ($(n1)!.4!(n2)$) [point, label = {above right:$a'$}] {};

\node (b) at ($(a)!.6!(n0)$) [point, label = left:$b$] {};
\node (b2) at ($(a2)!.6!(n0)$) [point, label = right:$b'$] {};
\node (z) at ($(a2)!.5!(n0)$) [point, label = right:$z$] {};

\node (A) at ($(n1)!.5!(n2)$) [label = {below:$A$}] {};

\draw[dotted, thick] (a) -- (n0);
\draw[dotted, thick] (a2) -- (n0);

\end{tikzpicture}
\end{center}
\caption{An illustration of the entities involved in the induction step.}
\end{figure}

Carath\'{e}odory's theorem (cf.~Schrijver~\cite[Section~7.7]{Schrijver:TLIP})
states that for every $p \in \conv(R)$, we can find $m+1 \leq d+1$ affinely independent points, 
$x_0, \dots, x_{m} \in R$, such that $p$ lies in $B = \conv(\{x_0, \dots, x_{m}\})$.
By induction over $m$, we show that for every point $b \in B$, and every $\epsilon > 0$,
there is a point $z \in R$ such that $\dist(b,z) < \epsilon$.
For $m = 0$, this statement follows trivially as each $x_j$ was chosen from $R$.
Now assume that $0 < m \leq d$ and that the statement holds for all $0 \leq m' < m$.
By the induction hypothesis, the statement holds for the set $A = \conv(\{x_0, \dots, x_{m-1}\})$.
Every $b \in B$ can be written as $b = y_b \cdot a + (1-y_b) \cdot x_{m}$ for some $a \in A$ and $0 \leq y_b \leq 1$.
Let $a' \in R$ be a point in $A$ that is at distance at most $\epsilon/2$ from $a$
and
let $b' = y_b \cdot a' + (1-y_b) \cdot x_{m}$.
Then, 
\begin{align*}
\dist(b, b') &= \| (y_b \cdot a + (1-y_b) \cdot x_{m}) - (y_b \cdot a' + (1-y_b) \cdot x_{m}) \|\\
& \leq y_b \| a-a' \|\\
& \leq \epsilon/2.
\end{align*}

Let $\delta > 0$ be a small constant to be fixed later.
By Lemma~\ref{lem:smalldelta1}, there exists a unary relation $U_{\delta} \in \{1\} + \mathcal{I}(\delta) \cap \langle R_+, U \rangle$.
Consider the following relation:
\[T(y) \equiv \exists t, z . U_{\delta}(t) \wedge R(z) \wedge z = y \cdot a' + (t-y) \cdot x_{m}.\]
Since $U_{\delta} \in \langle \{ R_+, U \} \rangle$, we also have $T \in \langle \{ R_+, U\} \rangle$.
By assumption, $T$ does not exclude an interval,
so there exists a $y'_b$ such that $T(y'_b)$ and $\left|y'_b - y_b\right| < \delta$.
Then, by the definition of $T$, there exists a $t \in (1-\delta,1+\delta)$ and a point $z \in R$ such that:
\begin{align*}
\dist(b', z)
&= \| \left(y_b \cdot a' + (1-y_b) \cdot x_{m}\right) - \left(y'_b \cdot a' + (t-y'_b) \cdot x_{m}\right) \|\\
&= \| (y_b-y'_b) \cdot a' + (1-t) \cdot x_{m} + (y'_b-y_b) \cdot x_{m} \|\\
&\leq \| (y_b-y'_b) \cdot a' \| + \| (1-t) \cdot x_{m} \| + \| (y'_b-y_b) \cdot x_{m} \|\\
&\leq \left( | y_b-y'_b | + | 1 - t | + | y'_b - y_b| \right) \max \{ \| a' \|, \| x_{m} \| \}\\
&< 3\delta C,
\end{align*}
where $C := \max \{\|a'\|, \|x_m\|, 1\}$ is a constant for a fixed $B$,
and the first inequality follows from the triangle inequality.
The claim now follows for the point $b$ by taking $\delta = (\epsilon/2) \cdot (3C)^{-1}$
since
$\dist(b,z) \leq \dist(b,b')+\dist(b',z) < \epsilon$.
\end{proof}

\begin{theorem}\label{thm:redtoessconv}
Let $\{R_+\} \subseteq \Gamma \subseteq SL_{\mathbb Q}[{\mathbb Q}]$ be a constraint
language.
Assume that there exists a unary relation $U \in \{1\} + \mathcal{I}(c) \cap \langle \Gamma \rangle$,
for some $0 < c < 1$, and
that every unary relation in $\langle \Gamma \rangle$ is essentially convex.
Then, CSP$(\Gamma)$ is equivalent to CSP$(\Gamma')$ for an essentially convex
constraint language $\Gamma' \subseteq SL_{\mathbb Q}[{\mathbb Q}]$.
\end{theorem}

\begin{proof}
If $\Gamma$ is essentially convex, then there is nothing to prove.
Assume therefore that $\Gamma$ is not essentially convex.
By Lemma~\ref{lemur}, each $R \in \Gamma$ can be defined by a formula
$\phi_0 \wedge \neg \phi_1 \wedge \dots \wedge \neg \phi_k$, where 
$\phi_0$, $\phi_1$, \dots, $\phi_k$
are conjunction over $LI_{\mathbb Q}[{\mathbb Q}]$, 
and $\phi_1$, \dots, $\phi_k$ define sets
whose affine hulls are of dimensions strictly lower than that of the set defined by $\phi_0$.
Assume additionally that the formulas are numbered so that the
affine hulls of the sets defined by $\phi_1, \dots \phi_m$ do not contain $(0,\dots,0)$ and
that the affine hulls of the sets defined by $\phi_{m+1}, \dots, \phi_k$  do contain $(0,\dots,0)$.
Define $R'$ by 
\[\phi \wedge \neg \phi'_1 \wedge \dots \wedge \neg \phi'_m \wedge \neg \phi_{m+1} \wedge \dots \wedge \neg \phi_k,\] 
where $\phi'_i$ defines the affine hull of the set defined by $\phi_i$.
Then, the constraint language $\Gamma' = \{ R' \mid R \in \Gamma \}$ is essentially convex
since witnesses of an excluded interval only occur inside an affine subspace not containing
$(0,\dots,0)$;
otherwise we could use such a witness to pp-define a unary relation excluding an interval.

Let $\Phi$ be an arbitrary instance of CSP$(\Gamma)$ over the variables $V=\{x_1,\ldots,x_n\}$ and 
assume $\Phi \equiv \exists x_1,\ldots,x_n.\psi$ where $\psi$ is quantifier-free.
Construct 
an instance $\Phi'$ of
CSP$(\Gamma')$ by replacing each occurrence of a relation $R$ in $\Phi$ by $R'$.
Clearly, if $\Phi'$ is satisfiable, then so is $\Phi$.
Conversely, let $s \in {\mathbb Q}^V$ be a solution to $\Phi$ and
assume that $\Phi'$ is not satisfiable.
Let $L$ be the line in ${\mathbb Q}^V$ through $(0,\dots,0)$ and $s$ and let $U$ be the unary relation 
$\mathcal{L}_{\psi,(0,\dots,0),s} \in \langle \Gamma \rangle$.
All tuples in $U$ correspond to solutions of $\Phi$ that are not solutions to $\Phi'$.

Fix a constraint $R(x_1, \dots, x_l)$ in $\Phi$ and consider the points in $U$ that satisfy this
constraint but not $R'(x_1, \dots, x_l)$.
These are the points $p \in {\mathbb Q}^V$ on $L$ for which $(p(x_1), \dots, p(x_l))$ satisfies $(\phi'_1 \vee \dots \vee \phi'_m) \wedge \neg (\phi_1 \vee \dots \vee \phi_m )$.
For each $1 \leq i \leq m$, $\phi'_i$ satisfies at most one point on $L$ since otherwise
the affine hull of the relation defined by $\phi_i$ would contain $(0,\dots,0)$.
Hence, each constraint in $\Phi$ can account for at most a finite number of points in $U$,
so $U$ is finite.

Assume first that $\left| U \right| > 1$.
Then, $U$ is not essentially convex which contradicts the assumption that
every unary relation in $\langle \Gamma \rangle$ is essentially convex.
Assume instead that $U = \{1\}$, where the single point in $U$ corresponds to the solution $s$.
Recall that $\Gamma$ is not essentially convex. Let $R \in \Gamma$ be a $k$-ary relation
that is not essentially convex and let $p,q \in {\mathbb Q}^k$ witness this. The relation
$\mathcal{L}_{R,p,q} \in \langle \Gamma \rangle$ since $\{1\} \in \langle \Gamma \rangle$.
Then, $\mathcal{L}_{R,p,q} \in \langle \Gamma \rangle$ is unary and not essentially convex
which leads to a contradiction.
It follows that if $\Phi$ is satisfiable, then so is $\Phi'$.
\end{proof}

\section{NP-hardness} \label{sec:hard}

We now derive a unified condition for all hard CSPs classified in this article.
It is based on a polynomial-time reduction from the NP-hard problem
Not-All-Equal 3SAT~\cite{Schaefer:stoc78}, i.e. the problem 
CSP$(\{R_{\rm NAE}\})$
where $R_{\rm NAE}=\{-1,1\}^3 \setminus \{(-1,-1,-1),(1,1,1)\}$.
The proof is divided into three different lemmas.
First, we present a reduction from 
Not-All-Equal 3SAT to a simple semilinear CSP.
We then show that having a bnu $T$ that is bounded away from $0$ allows us to
pp-define unary relations that are, in a certain sense, close to being either
the relation $\{1\}$ or $\{-1,1\}$. In the final step, we combine these two results
and show that having a bnu $T$ that excludes an interval
and that is bounded away from $0$ is a sufficient condition for CSP($\{R_+, T\}$)
to be NP-hard. 
\begin{lemma} \label{plusminushardness}
Let $T \in \{-1,1\} + \mathcal{I}(\frac{1}{2})$.
Then,
CSP$(\{R_+,T\})$ is NP-hard.
\end{lemma}
\begin{proof}
The proof is by a polynomial time reduction from
CSP$(\{R_{\rm NAE}\})$.
Let
$\Phi$ denote an arbitrary instance of CSP$(\{R_{\rm NAE}\})$. Construct an instance $\Phi'$
of CSP$(\{R_+,T\})$ as follows.
Impose the constraint $T(v)$ on each variable. For each
constraint $R_{\rm NAE}(x,y,z)$ in $\Phi$, introduce the constraints
$x+y+z+w=0$ and $T(w)$,
where $w$ is a fresh variable.

Assume that $\Phi$ has a solution. Consider a constraint $R_{\rm NAE}(x,y,z)$ in $\Phi$.
If two of the variables are assigned the value $1$, then
the equation $x+y+z+w=0$ is satisfied by choosing $w=-1$.
If two of the variables are assigned the value $-1$, then
the equation $x+y+z+w=0$ is satisfied by choosing $w=1$.
Hence, $\Phi'$ is satisfiable.

Assume that $\Phi'$ has a solution $s'$. Then, $\Phi$ has a solution $s$ defined
by $s(x)=1$ if $s'(x)>0$ and $s(x)=-1$ if $s'(x)<0$.
Assume to the contrary that $s(x) = s(y) = s(z) = 1$ for some variables
with a constraint $R_{\rm NAE}(x,y,z)$.
Consider the equation $x+y+z+w=0$ in $\Phi'$.
By the assumption on $T$, we have
$s'(x)+s'(y)+s'(z) > \frac{3}{2}$, and hence $s'(w) <-\frac{3}{2}$.
But this is a contradiction as the constraint $T(w)$ is also in $\Phi'$.
We can similarly rule out the case
$s(x)=s(y)=s(z)=-1$. This proves that $s$ is a solution to $\Phi$.
\end{proof}

\begin{lemma} \label{lem:thin}
Let $T \neq \emptyset$ be a bounded unary relation such that $T \cap (-\epsilon,\epsilon) = \emptyset$,
for some $\epsilon > 0$.
Then,
either
$\langle R_+, T \rangle$ contains a unary relation $U_{\delta} \in \{1\} + \mathcal{I}(\delta)$
for every $\delta > 0$; or
$\langle R_+, T \rangle$ contains a unary relation $U_{\delta} \in \{-1,1\} + \mathcal{I}(\delta)$,
for every $\delta > 0$.
\end{lemma}

\begin{proof}
If $T \cap -T \neq \emptyset$, then the result follows from Lemma~\ref{lem:smalldelta2}.
Otherwise, by Lemma~\ref{fakeomin}, there exists a constant $c^+ > 0$ such that
the set $T^+ = \{ x \in T \mid |x| \geq c^+ \}$ is 
non-empty and contains points that are either all positive or all negative.
Similarly, there exists a constant $c^- > 0$ such that
$T^- = \{ x \in T \mid |x| \leq c^- \}$ is non-empty and contains points that are 
either all positive or all negative.
Let $a \in T^+$ and $b \in T^-$.
Assume that both sets contain positive points only or that both sets contain negative points only.
Then, the result follows using Lemma~\ref{lem:smalldelta1} with the relation
$U = a^{-1} \cdot T \cap b^{-1} \cdot T$ (or $-U$ if the points of $U$ are negative).
The case when the one set contains positive points and the other contains negative points
is handled similarly using the relation $U' = a^{-1} \cdot T \cap b^{-1} \cdot (-T)$.
\end{proof}

\begin{lemma} \label{universalhardness}
Let $T$ be a bnu such that $T \cap (-\epsilon,\epsilon) = \emptyset$,
for some $\epsilon > 0$, 
and $U$ be a unary relation that excludes an interval.
Then,
CSP$(\{R_+,T, U\})$ is NP-hard.
\end{lemma}

\begin{proof}
We show that
$\langle R_+, T, U \rangle$ contains a unary relation $\{-1,1\}+ \mathcal{I}(\frac{1}{2})$.
The result then follows from Lemma~\ref{plusminushardness}.
If already $\langle R_+, T \rangle$ contains such a relation, then we are done.
Otherwise, by Lemma~\ref{lem:thin}, $\langle R_+, T \rangle$ contains a unary relation
$U_{\delta} \in \{1\} + \mathcal{I}(\delta)$, for every $\delta > 0$.
Since $U$ excludes an interval, there are points
$p, q \in U$ and $0 < \delta_1 < \delta_2 < 1$ such that $p+(q-p)y \not\in U$ whenever
$\delta_1 \leq y \leq \delta_2$.
Furthermore, $p$ and $q$ can be chosen so that $\delta_1 < 1/2 < \delta_2$,
and by scaling $U$, we may assume that $|q-p| = 2$.
Let $m = (p+q)/2$.
Note that $T \cap (m-\epsilon',m+\epsilon') = \emptyset$, for some $\epsilon' > 0$.
Similarly, possibly by first scaling $T$, let $p', q' \in T$ be distinct points with $|q'-p'| = 2$
and let $m' = (p'+q')/2$.

Now, define the following unary relations:
\begin{align*}
T_{0}(x) &\equiv
\exists y \exists z \: . \: 
U_\delta(y) \wedge z = x - y \cdot m \wedge U(z) \\
T_{\infty}(x) &\equiv
\exists y' \exists z' \: . \: 
U_\delta(y') \wedge z' = x - y' \cdot m' \wedge T(z').
\end{align*}
The relations $T_0$ and $T_{\infty}$ are roughly translations of $U$ and $T$,
where the constant relation $\{1\}$ has been approximated by the relation $U_\delta$.
Since $1 \in U_\delta$, we have $\{-1, 1\} \subseteq T_0, T_\infty$.
Hence, if $\delta$ is chosen small enough, then the relation $T_0 \cap T_\infty \in \langle R_+, T, U \rangle$ will satisfy the conditions of Lemma~\ref{lem:smalldelta2}.
This finishes the proof.
\end{proof}

\section{Semilinear expansions of $\{R_+\}$} \label{sec:mainresult}

In this section, we prove our main result: Theorem~\ref{superresult}.
We divide the proof into two parts.
Consider the following two properties:

\begin{itemize}
\item[(P$_{0}$)] There is a unary relation $U$ in $\langle \Gamma \rangle$ that contains a positive point and satisfies $U \cap (0,\epsilon) = \emptyset$ for some $\epsilon > 0$.
\item[(P$_{\infty}$)] There is a unary relation $U$ in $\langle \Gamma \rangle$ that contains a positive point and satisfies $U \cap (M,\infty) = \emptyset$ for some $M < \infty$.
\end{itemize}

In the first part of the proof (Section~\ref{property-classification}), we consider
constraint languages that simultaneously satisfy the properties (P$_0$) and (P$_{\infty}$).
In the second part (Section~\ref{noproperty-classification}), we consider
constraint languages that violates at least one of them.
In both parts, we give a detailed description of the boundary between easy and hard problems.
By combining Theorem~\ref{thm:main} and Theorem~\ref{nopropertymainresult}, we establish
Theorem~\ref{superresult}.

In addition to the two algorithmic results in Section~\ref{sec:affine}, 
there is also a trivial source of tractability.
A relation is \emph{$0$-valid} if it contains the tuple $(0,\dots,0)$ and
a constraint language is \emph{$0$-valid} if every relation in it is $0$-valid.
Every instance of a CSP over a $0$-valid constraint language admits the solution that
assigns $0$ to every variable.

When we consider constraint languages that are not $0$-valid,
the following lemma shows that there is always a pp-definable 
unary relation that is not $0$-valid.

\begin{lemma} \label{zerovalidlemma}
Let $\{R_+\} \subseteq \Gamma \subseteq SL_{\mathbb Q}[{\mathbb Q}]$ be a constraint language.
If $\Gamma$ is not $0$-valid,
then $\langle \Gamma \rangle$ contains a non-empty unary relation that is not $0$-valid.
\end{lemma}
\begin{proof}
By assumption, $\Gamma$ contains some $k$-ary relation $R$ that is not $0$-valid,
and by our definition of a constraint language, $R$ is non-empty.
Let $t \in R$ be a tuple that contains the largest possible
number $m$ of zeroes. Assume for simplicity that the first $m$ entries
of $t$ equals $0$. Consider the following unary relation in $\langle \Gamma \rangle$.
\[U = \{x \in {\mathbb Q} \mid \exists y_{m+1} \dots y_{k-1} \: . \: R(0,0,\dots,0,y_{m+1},\dots,y_{k-1},x)\}\]
The relation $U$ is non-empty and not $0$-valid.
\end{proof}

\subsection{The case (P$_0$) and (P$_{\infty}$)}
\label{property-classification}

The following theorem covers the case when the constraint language satisfies both
of the properties (P$_0$) and (P$_{\infty}$).
As a corollary, we obtain a complete classification for semilinear constraint languages
containing $\{R_+,\{1\}\}$.
The latter result
is interesting in itself and it will also be used
in Section~\ref{noproperty-classification} and Section~\ref{sec:optresult}.

\begin{theorem}\label{thm:main}
Let $\{R_+\} \subseteq \Gamma \subseteq SL_{\mathbb Q}[{\mathbb Q}]$ be a finite constraint language that satisfies (P$_0$) and (P$_\infty$).
The problem CSP$(\Gamma)$ is in P if
\begin{itemize}
\item
$\Gamma$ is 0-valid (trivially);
\item
$\langle \Gamma \rangle$ does not contain a bnu (by establishing affine consistency); or
\item
all unary relations in $\langle \Gamma \rangle$ are essentially convex (by a reduction to an essentially convex constraint language). 
\end{itemize}
Otherwise, CSP$(\Gamma)$ is NP-hard.
\end{theorem}

\begin{proof}
   Let $\mathcal{U}$ be the set of all bounded, non-empty unary relations $U$ in $\langle \Gamma \rangle$ such that $U \cap (-\epsilon, \epsilon) = \emptyset$ for some $\epsilon > 0$.
   Assume that $\Gamma$ is not $0$-valid.
   First, we show that $\mathcal{U}$ is non-empty.
   By Lemma~\ref{zerovalidlemma}, $\langle \Gamma \rangle$ contains a non-empty unary
   relation that is not $0$-valid. Scale this relation so that it contains $1$ and call the
   resulting relation $U'$.
   Let $U_0 \in \langle \Gamma \rangle$ be a unary relation witnessing (P$_0$) and
   let $U_\infty \in \langle \Gamma \rangle$ be a unary relation witnessing (P$_\infty$).
   Scale $U_0$ and $U_\infty$ so that some positive point from each coincides with 1 and let
   $T = U' \cap U_0 \cap U_\infty$.
   If $T$ does not contain a negative point, then $T \in \mathcal{U}$.
   Otherwise, $T$ contains a negative point $b$.
   It follows that $T \cap b\cdot T \in \mathcal{U}$.
   Hence, the set $\mathcal{U}$ is non-empty.
   
   Assume that $\langle \Gamma \rangle$ does not contain a bnu.
   Then, neither does $\mathcal{U}$ and hence $\mathcal{U}$ contains only constants.
   It follows by Theorem~\ref{thm:affinesolves} that establishing affine consistency solves
   CSP($\Gamma$).
      
   Otherwise, $\mathcal{U}$ contains a bnu.
   If all unary relations of $\langle \Gamma \rangle$ are essentially convex, then
   by Lemma~\ref{lem:thin} and
   Theorem~\ref{thm:redtoessconv}, CSP$(\Gamma)$ is equivalent to CSP$(\Gamma')$
   for an essentially convex constraint language $\Gamma'$.
   Tractability follows from Theorem~\ref{thm:essconvtractable}.
   
   Finally, if $\mathcal{U}$ contains a bnu and $\langle \Gamma \rangle$ contains a
   unary relation that excludes an interval, then NP-hardness follows from
   Lemma~\ref{universalhardness}.
\end{proof}

\begin{corollary} \label{linearmainresult}
Let $\{R_+,\{1\}\} \subseteq \Gamma 
\subseteq SL_{\mathbb Q}[{\mathbb Q}]$ 
be a finite constraint language.
The problem CSP$(\Gamma)$ is in P if
$\langle \Gamma \rangle$ does not contain a bnu or if $\Gamma$ is essentially convex.
Otherwise, CSP$(\Gamma)$ is NP-hard.
\end{corollary}

\begin{proof}
   If $\langle \Gamma \rangle$ does not contain a bnu, then tractability follows
   from Theorem~\ref{thm:affinesolves}.
   If all relations in $\Gamma$ are essentially convex, then tractability follows
   from Theorem~\ref{thm:essconvtractable}.
   
   Otherwise, $\langle \Gamma \rangle$ contains a bnu, and
   since $\{R_+, \{1\}\} \subseteq \Gamma$, 
   $\langle \Gamma \rangle$ also contains a unary relation that is not essentially convex.
   Since $\{1\} \in \Gamma$ is not $0$-valid,
   NP-hardness then follows from Theorem~\ref{thm:main}.
\end{proof}

\subsection{The case $\neg ($P$_0)$ or $\neg ($P$_\infty)$}
\label{noproperty-classification}

Let $\{R_+\} \subseteq \Gamma \subseteq SL_{\mathbb Q}[{\mathbb Q}]$ be a constraint language
such that either (P$_0$) or (P$_{\infty}$) is violated.
In this section, we show that $\Gamma$ can be replaced by an equivalent constraint
language of a restricted type.
Let $HSL_{\mathbb Q}[{\mathbb Q}]$ denote the set of relations that are
are finite unions of homogeneous linear sets.
We will call such relations \emph{homogeneous} semilinear relations.
We remind the reader that we can always pp-define the relations $\{0\}$ and
$M=\{(x,-x) \mid x \in {\mathbb Q}\}$ in $\Gamma$: $x=0 \Leftrightarrow R_+(x,x,x)$
and $(x,y) \in M \Leftrightarrow R_+(x,y,0) \Leftrightarrow \exists z.R_+(x,y,z) \wedge R_+(z,z,z)$. 
Hence, we can 
freely use the constant 0 and negation in forthcoming pp-definitions.

From now on, let
${\mathbb Q}_+=\{a \in {\mathbb Q} \; | \; a > 0\}$,
${\mathbb Q}_-=\{a \in {\mathbb Q} \; | \; a < 0\}$, and
${\mathbb Q}_{\neq 0} = {\mathbb Q}_- \cup {\mathbb Q}_+ = {\mathbb Q} \setminus \{0\}$.
For a relation $R \in SL_{\mathbb Q}[{\mathbb Q}]$,
     define 
     $\cone(R) = \{ \lambda \cdot x \mid \lambda \in {\mathbb Q_+}, x \in R \}$
     to be the \emph{cone} over $R$.
For a constraint language $\Gamma \subseteq SL_{\mathbb Q}[{\mathbb Q}]$,
let $\cone(\Gamma) = \{ \cone(R) \mid R \in \Gamma \}$.
Note that,
for $\{R_+\} \subseteq \Gamma \subseteq SL_{\mathbb Q}[{\mathbb Q}]$,
we have $\cone(\Gamma) \subseteq HSL_{\mathbb Q}[{\mathbb Q}]$, and
since $\cone(R_+) = R_+$, we also have $R_+ \in \cone(\Gamma)$.

For an assignment $s : V \to {\mathbb Q}$ and a rational $c \in {\mathbb Q}$, let
$c \cdot s$ denote the assignment $x \mapsto c \cdot s(x)$.

\begin{theorem} \label{thm:mainreduction}
Let 
$\{R_+\} \subseteq \Gamma \subseteq SL_{\mathbb Q}[{\mathbb Q}]$
be a constraint language such that either (P$_0$) or (P$_\infty$) is violated.
Then, CSP$(\Gamma)$ is equivalent to CSP$(\cone(\Gamma))$.
\end{theorem}

\begin{proof} 
     Assume that $\Gamma$ does not satisfy (P$_0$).
     The proof for the case when $\Gamma$ does not satisfy (P$_\infty$)
follows
     similarly.
     
     Let $R$ be a relation in $\Gamma$ and
     let $\phi = \phi_1 \vee \dots \vee \phi_k$ 
     be a quantifier-free
     DNF formula for $R$, where each formula $\phi_j$ is conjunction of
     strict and non-strict inequalities.
     Remove every disjunct $\phi_j$ that contains a non-homogeneous
inequality which is not
     satisfied by the $(0,\dots,0)$-tuple.
     Let $S$ be the relation defined by the resulting formula $\phi' =
\phi'_1 \vee \dots \vee \phi'_{k'}$.
     Since $\Gamma$ does not satisfy (P$_0$), it follows that for every
point $x$ in $R \setminus S$,
     there is a point $x'$ in $S$ that lies on the open line segment between
$(0,\dots,0)$ and $x$.
     Therefore, $\cone(S) = \cone(R)$.
     Next, for each $j$, let $S_j$ be the relation defined by $\phi'_j$.
     Remove every non-homogeneous inequality from $\phi'_j$, let
$\phi''_j$ be the resulting formula
     and let $T_j$ be the relation defined by $\phi''_j$.
     Clearly, $\cone(S_j) \subseteq \cone(T_j)$.
     Let $\lambda \cdot x$ be a point in $\cone(T_j)$ with $\lambda \in
{\mathbb Q}_+$ and $x \in T_j$.
     Since every non-homogeneous inequality in $\phi'_j$ is satisfied by
the $(0,\dots,0)$-tuple, it follows
     that they are satisfied by every point in a small ball $B$ centred
at $(0,\dots,0)$.
     Let $x'$ be a point in $B$ on the line segment between $(0,\dots,0)$ and $x$ and
     note that every homogeneous inequality in $\phi'_j$ satisfies $x$
and therefore also $x'$.
     It follows that $x'$ is in $S_j$ so $x$ and $\lambda \cdot x$ are
in $\cone(S_j)$,
     which shows that $\cone(T_j) \subseteq \cone(S_j)$.
     Let $\phi'' = \phi''_1 \vee \dots \vee \phi''_{k'}$ and let
     $T$ be the relation defined by $\phi''$.
     Then, $\cone(R) = \cone(T)$ and $\cone(T) = T$ since $\phi''$ only
contains homogeneous inequalities.
     Therefore, $\phi''$ defines $\cone(R)$, so $\cone(R) \in HSL_{\mathbb
Q}[{\mathbb Q}]$
     and $\Gamma' \subseteq HSL_{\mathbb Q}[{\mathbb Q}]$.

     For the equivalence of CSP$(\Gamma)$ and CSP$(\Gamma')$,
     arbitrarily choose an instance $\Phi$ of CSP$(\Gamma)$.
     Construct an instance $\Phi'$ of CSP$(\Gamma')$ by
     replacing each occurrence of a relation $R$ in $\Phi$ by $\cone(R)$.
     Every solution to $\Phi$ is also a solution to $\Phi'$.
     It remains to show that if $\Phi'$ has a solution, then so does $\Phi$.

     Let $s : \vars(\Phi') \to {\mathbb Q}$ be a solution to $\Phi'$ and assume without
loss of generality that $s$ is integral.
     If $s \equiv 0$, then it follows immediately that $s$ is a solution to
$\Phi$ since,
for every $R \in \Gamma$,
 $(0,\dots,0) \in \cone(R)$ if and only if $(0,\dots,0) \in R$.
 Assume therefore that $s \not\equiv 0$.
     Let $R_1(x_1), \dots, R_m(x_m)$ be the constraints of $\Phi$.
     For every $j$, $s(x_j) \in \cone(R_j)$ holds.
     By the construction of $\cone(R_j)$, this implies that $r \cdot
s(x_j) \in
R_j$, for some $r > 0$.
     Define the unary relation $U \in \langle \Gamma \rangle$ by the
pp-formula
     $\psi(y) \equiv \exists x \: . \:  x = y \cdot s(x_j) \wedge R_j(x)$.
     Now $r \in U$, so by the assumption on $\Gamma$, it follows that
$(0, \epsilon_j) \subseteq U$,
     for some $\epsilon_j > 0$, and hence that $y \cdot s(x_j) \in R_j$,
for all $y \in (0, \epsilon_j)$.
     Let $\epsilon = \min_j {\epsilon_j}$.
     Then $(\epsilon/2) \cdot s$ is a solution to $\Phi$.
\end{proof}

By Theorem~\ref{thm:mainreduction},
it is thus sufficient to determine 
the computational complexity of CSP$(\Gamma)$ for
$\{R_+\} \subseteq \Gamma \subseteq HSL_{\mathbb Q}[{\mathbb Q}]$.

Given a relation $R \subseteq {\mathbb Q}^k$, we say that a function $e:{\mathbb Q} \rightarrow {\mathbb Q}$
is an {\em endomorphism} of $R$ if for every tuple $(a_1,\dots,a_k) \in R$, the
tuple $(e(a_1),\dots,e(a_k)) \in R$. One may equivalently view an endomorphism
as a homomorphism from $R$ to $R$. We extend this notion to constraint languages $\Gamma=\{R_1,\dots,R_n\}$: 
a function $e:{\mathbb Q} \rightarrow {\mathbb Q}$ is an endomorphism of $\Gamma$
if $e$ is an endomorphism of $R_i$, $1 \leq i \leq n$.

\begin{lemma} \label{endolemma}
Let $a > 0$ be a rational number.
Every $R \in HSL_{\mathbb Q}[{\mathbb Q}]$ has the endomorphism $e(x)=a \cdot x$.
\end{lemma}
\begin{proof}
We know that $R$ can be written as
$R = \bigcup_{i=1}^m H_i$ where $H_i$, $1 \leq i \leq m$, is defined by a (finite) system
of homogeneous linear (strict or non-strict) inequalities. Consider an inequality
$\sum_{i=1}^n c_i \cdot x_i \geq 0$ in such a system.
We immediately see that 
\[\sum_{i=1}^n c_i \cdot x_i \geq 0 \Leftrightarrow
a \cdot \sum_{i=1}^n c_i \cdot x_i \geq 0 \Leftrightarrow
\sum_{i=1}^n a \cdot c_i \cdot x_i \geq 0 \Leftrightarrow
\sum_{i=1}^n c_i \cdot e(x_i) \geq 0.\]
This equivalence also holds if we consider strict inequalities.
Therefore, each $H_i$, $1 \leq i \leq m$, has the endomorphism $e$.

Now, arbitrarily choose
a tuple
$t = (t_1,\ldots,t_k) \in R$ and assume that $t \in H_i$.
It follows that $(e(t_1),\ldots,e(t_k)) \in H_i \subseteq R$,
so the function $e$ is an endomorphism of $R$.
\end{proof}

A direct consequence of Lemma~\ref{endolemma} is the following: if an instance $\Phi$
of CSP$(HSL_{\mathbb Q}[{\mathbb Q}])$ has a solution $s$, then $a \cdot s$ is a solution
for every rational number $a > 0$.

The complexity classification of constraint languages that violate either
$(P_0)$ or $(P_\infty)$, in Theorem~\ref{nopropertymainresult}, 
follows from two intermediate results
which we now present in Lemma~\ref{pointlemma} and Lemma~\ref{homogeneouscaselemma}.

\begin{lemma} \label{pointlemma}
Let $\Gamma$ be a subset of $HSL_{\mathbb Q}[{\mathbb Q}]$ and let 
$U$ be a unary relation in $\langle \Gamma \rangle$. If $U$ contains
an element $p > 0$, then ${\mathbb Q}_+ \subseteq U$.
If $U$ contains
an element $p < 0$, then ${\mathbb Q}_- \subseteq U$.
\end{lemma}
\begin{proof}
Let $q \in {\mathbb Q}$ be any element with the same sign as $p$.
By Lemma~\ref{endolemma}, $e(x) = (q/p) \cdot x$ is an endomorphism of $U$.
Since $p \in U$, it follows that $q = e(p) \in U$.
\end{proof}

\begin{lemma} \label{homogeneouscaselemma}
Let $\{R_+\} \subseteq \Gamma \subseteq HSL_{\mathbb Q}[{\mathbb Q}]$ be a finite constraint language.
Then, either
\begin{itemize}
\item
$\Gamma$ is $0$-valid; or
\item
CSP$(\Gamma)$ is polynomial-time equivalent to CSP$(\Gamma \cup \{\{1\}\})$.
\end{itemize}
\end{lemma}

\begin{proof}
Assume that $\Gamma$ is not $0$-valid.
By Lemma~\ref{zerovalidlemma}, $\langle \Gamma \rangle$ contains a non-empty unary relation that
is not $0$-valid.
We consider 
three different cases.
\medskip

\noindent
{\em Case 1.} $\langle \Gamma \rangle$ contains a non-empty unary relation $U$ such that $0 \not\in U$
and $U \subseteq {\mathbb Q}_+$. By Lemma~\ref{pointlemma}, ${\mathbb Q}_+ \subseteq U$ so
$U={\mathbb Q}_+$. We claim that
CSP$(\Gamma \cup \{\{1\},{\mathbb Q}_+\})$ is polynomial-time equivalent to
CSP$(\Gamma \cup \{{\mathbb Q}_+\})$. The polynomial-time reduction from right to left
is trivial. To show the other direction, let $\Phi$ be an arbitrary instance of
CSP$(\Gamma \cup \{\{1\},{\mathbb Q}_+\})$. Assume without loss of generality that
the relation $\{1\}$ appears in exactly one constraint $\{1\}(x)$. Construct $\Phi'$
by replacing this constraint with ${\mathbb Q}_+(x)$.

If $\Phi'$ has no solution, then $\Phi$ has no solution.
Suppose instead that $\Phi'$ has the solution $s$. Then we know
that $s(x) > 0$. Choose $a \in {\mathbb Q}$ such that $a \cdot s(x) = 1$.
By Lemma~\ref{endolemma},
the function $a \cdot s$ is then a solution to $\Phi$.
\medskip

\noindent
{\em Case 2.} $\langle \Gamma \rangle$ contains a non-empty unary relation $U$ such that $0 \not\in U$
and $U \subseteq {\mathbb Q}_-$. By Lemma~\ref{pointlemma}, ${\mathbb Q}_- \subseteq U$ so
$U={\mathbb Q}_-$. We can now pp-define ${\mathbb Q}_+$ since
$x > 0 \Leftrightarrow -x < 0$ and go back to Case 1.
\medskip

\noindent
{\em Case 3.} $\langle \Gamma \rangle$ contains a non-empty unary relation $U$ such that $0 \not\in U$
and no unary relation $U' \in \langle \Gamma \rangle$ equals ${\mathbb Q}_+$ or ${\mathbb Q}_-$.
Lemma~\ref{pointlemma} implies that $U={\mathbb Q}_- \cup {\mathbb Q}_+$. 

We claim that
CSP$(\Gamma)$ is polynomial-time equivalent to
CSP$(\Gamma \cup \{\{1\}\})$.
The reduction from left to right is trivial.
To show the other direction, let $\Phi\equiv\exists x_1,\dots,x_m \: . \: \phi(x_1,\dots,x_m)$
be an arbitrary instance of CSP$(\Gamma \cup \{\{1\},U\})$,
where $\phi$ is quantifier-free, and
assume without loss of generality that
the relation $\{1\}$ appears in exactly one constraint $\{1\}(x_m)$. Construct $\Phi'$
by replacing this constraint with ${\mathbb Q}_{\neq 0}(x_m)$.

If $\Phi'$ has no solution, then $\Phi$ has no solution.
Suppose instead that $\Phi'$ has a solution. 
Assume first that every solution assigns a negative number to the variable $x_m$.
Then we can pp-define a unary relation $T \subseteq {\mathbb Q}_-$ by
\[T(x_m) \equiv \exists x_1,\dots,x_{m-1} \: . \: \phi(x_1,\dots,x_m)\]
and this contradicts our initial assumptions. Thus, there is a solution $s$
such that $s(x_m) > 0$.
Choose $a \in {\mathbb Q}$ such that $a \cdot s(x) = 1$.
By Lemma~\ref{endolemma}, the function $a \cdot s$ is a solution to $\Phi$.
\end{proof}

\begin{theorem} \label{nopropertymainresult}
Let $\{R_+\} \subseteq \Gamma \subseteq SL_{\mathbb Q}[{\mathbb Q}]$ be a finite constraint language that violates (P$_0$) and/or (P$_{\infty}$).
The problem CSP$(\Gamma)$ is in P if
\begin{itemize}
\item
$\Gamma$ is $0$-valid;
\item
$\langle \cone(\Gamma) \cup \{\{1\}\} \rangle$ does not contain a bnu; or
\item
$\cone(\Gamma)$ is essentially convex.
\end{itemize}
Otherwise, CSP$(\Gamma)$ is NP-hard.
\end{theorem}

\begin{proof}
By Theorem~\ref{thm:mainreduction}, CSP$(\Gamma)$ is equivalent to CSP$(\cone(\Gamma))$.
By Lemma~\ref{homogeneouscaselemma}, CSP$(\cone(\Gamma))$ is either trivially in P,
if it is $0$-valid, or CSP$(\cone(\Gamma))$ is polynomial-time equivalent to CSP$(\cone(\Gamma) \cup \{\{1\}\})$.
In the latter case, the result follows from Corollary~\ref{linearmainresult}.
\end{proof}

\section{Optimisation}
\label{sec:optresult}

In this section, we study the optimisation problem where the objective is
to maximise a linear function over the solution set of a semilinear CSP.
For an arbitrary constraint language $\Gamma \subseteq SL_{\mathbb Q}[{\mathbb Q}]$,
we formally define the problem
Opt$(\Gamma)$ as follows.

\bigskip
\begin{center}
\fbox{
  \parbox{0.9\textwidth}{
{\bf Problem:} Opt$(\Gamma)$

\noindent
{\bf Input:} 
A CSP$(\Gamma)$-instance $\Phi$ and a vector $c \in {\mathbb Q}^{\vars(\Phi)}$.

\noindent
{\bf Output:} One of the following four answers.

\begin{itemize}
\item
`unbounded' if for every $K \in {\mathbb Q}$, there exists a solution $x$ such that $c^Tx \geq K$.

\item
`optimum: $K$' if there exists a $K \in {\mathbb Q}$ and a solution $x$ such that $c^Tx = K$, but there
is no solution $x'$ such that $c^Tx' > K$.

\item
`optimum is arbitrarily close to $K$' if there exists a $K \in {\mathbb Q}$ such that there is
no solution $x$ satisfying $c^Tx \geq K$, but for every $K' < K$ there is a solution $x'$ with
$c^Tx' \geq K'$.

\item
`unsatisfiable' if there is no solution.
\end{itemize}
}
}
\end{center}
\bigskip

By Lemma~\ref{generateequations}, the problem Opt$(\{R_+, \leq, \{1\})$ is
polynomial-time equivalent to linear programming.
Bodirsky et al.~\cite{Bodirsky:etal:lmcs2012} have shown
that for semilinear constraint languages containing $\{R_+, \leq, \{1\}\}$,
the problem CSP$(\Gamma)$ is polynomial-time solvable (NP-hard) if and only if
the problem Opt$(\Gamma)$ is polynomial-time solvable (NP-hard)
(cf.~Theorem~\ref{bodirskytheorem}).

In Theorem~\ref{thm:main-opt},
we show that, for semilinear constraint languages containing $\{R_+, \{1\}\}$,
the complexity of the decision problem and of the optimisation problem is similarly related.
We first prove an analogue of Theorem~\ref{thm:affinesolves}
for the optimisation problem.

\begin{theorem}
\label{thm:affine-opt}
Let $\{R_+, \{1\}\} \subseteq \Gamma \subseteq SL_{\mathbb Q}[{\mathbb Q}]$ be a finite constraint language.
If there is no bnu in $\langle \Gamma \rangle$, then Opt$(\Gamma)$
can be solved in polynomial time.
\end{theorem}
\begin{proof}
Let $\Phi$ be an instance of CSP($\Gamma$), let $V=\vars(\Phi)=\{x_1,\ldots,x_m\}$, and let $c \in {\mathbb Q}^V$ be a vector. Assume $\Phi\equiv\exists x_1,\dots,x_m \: . \: \phi$ where $\phi$ is
quantifier-free.
Algorithm~\ref{alg:affc} in Section~\ref{sec:tractability}
finds the affine hull $A$ of the set of satisfying assignments to $\Phi$
in polynomial time.
If $A = \emptyset$, then we answer `unsatisfiable'.

Otherwise, the affine hull $A$ is represented by a set of inequalities,
each with representation size that is polynomial in the input size.
Therefore,
we can solve the system $z_1, z_2 \in A$, $c^T (z_1-z_2) > 0$,
in polynomial time.
Assume that this system has a solution.
Let $k = \dim(A)+1$ and let
$y_1, \dots, y_k$ be affinely independent satisfying assignments to $\Phi$.
Then, we can write
$z_1 = \sum_{i=1}^k a_{1i} y_i$ and
$z_2 = \sum_{i=1}^k a_{2i} y_i$ with $\sum_{i=1}^k a_{1i} = \sum_{i=1}^k a_{2i} = 1$.
Since 
\[c^T (z_1-z_2) = \sum_{i=1}^k a_{1i} c^T y_i - \sum_{i=1}^k a_{2i} c^T y_i > 0,\]
we must have $c^T y_i \neq c^T y_j$ for some $1 \leq i,  j \leq k$.
Let $U = {\mathcal L}_{R_\phi, y_i, y_j} \in \langle \Gamma \rangle$,
where $R_\phi = \{ (x_1, \dots, x_n) \in {\mathbb Q}^V \mid \phi(x_1,\dots,x_n) \text{ is true in } \Gamma \}$
 and
for each $a \in U$, let $y_a \in {\mathbb Q}^V$ 
denote the corresponding point on the line through $y_i$ and $y_j$.
Fix an arbitrary constant $K \in {\mathbb Q}$.
Since there is no bnu in $\langle \Gamma \rangle$, it follows from Lemma~\ref{halfboundgenerate}(1) 
that there is a point $a \in U$ such that $y_a \in {\mathbb Q}^V$ satisfies $c^T y_a > K$.
Since $y_a$ is a satisfying assignment, we can therefore answer `unbounded'.

Otherwise, $c^T(z_1-z_2) = 0$ for all $z_1, z_2 \in A$, so
$c^T z$ is constant for $z \in A$.
In polynomial time, we can find a $z \in A$ with polynomial representation size.
It then suffices to evaluate $c^T z$ and answer `optimum: $c^T z$'.
\end{proof}

\begin{theorem}
\label{thm:main-opt}
Let $\{R_+,\{1\}\} \subseteq \Gamma 
\subseteq SL_{\mathbb Q}[{\mathbb Q}]$ 
be a finite constraint language.
The problem Opt$(\Gamma)$ is polynomial-time solvable if
$\langle \Gamma \rangle$ does not contain a bnu or if
$\Gamma$ is essentially convex.
Otherwise, Opt$(\Gamma)$ is NP-hard.
\end{theorem}

\begin{proof}
The polynomial-time solvable cases follow from Theorem~\ref{bodirskytheorem} and Theorem~\ref{thm:affine-opt}.
The hardness follows from Corollary~\ref{linearmainresult}.
\end{proof}

A comparison between Theorem~\ref{thm:main-opt} and Corollary~\ref{linearmainresult} shows that,
for a semilinear constraint language $\Gamma$ containing $\{R_+, \{1\}\}$,
CSP$(\Gamma)$ is polynomial-time solvable (NP-hard) if and only if Opt$(\Gamma)$
is polynomial-time solvable (NP-hard).
The following example shows that this tight relationship between the complexity of 
a constraint satisfaction problem and its corresponding optimisation problem
cannot be further extended
to the class of all semilinear constraint languages containing the relation $R_+$.

\begin{example}
Let $R = \{(0,0,0,0)\} \cup \{(x,y,z,1) \mid (x,y,z) \in R_{{\rm NAE}}\}$
(cf.~Section~\ref{sec:hard}).
Note that $\Gamma = \{R, R_+\}$ is semilinear, $0$-valid, and that
$\Gamma$ satisfies both (P$_0$) and (P$_{\infty}$).
Let $\Phi$ be an arbitrary instance of CSP$(\{R_{\rm NAE}\})$.
Construct an instance $\Phi'$ of Opt$(\Gamma)$ by introducing an auxiliary
variable $w$, and for each constraint $R_{\rm NAE}(x,y,z)$ in $\Phi$,
introduce a constraint $R(x,y,z,w)$ in $\Phi'$.
Finally, let the vector $c \in {\mathbb Q}^{\vars(\Phi')}$ be defined by
$c_w = 1$ and $c_x = 0$ for all other variables $x$.
Then, the instance $\Phi$ has a solution if and only if an optimal solution of $\Phi'$ has
value 1.
We conclude that CSP$(\Gamma)$ is polynomial-time solvable (since $\Gamma$ is $0$-valid),
but that Opt$(\Gamma)$ is NP-hard.
\end{example}

\section{Integer solutions}
\label{sec:intsol}

In this section, we study the problem of finding \emph{integer} solutions to CSPs defined
by semilinear relations.
We consider two different approaches: 
(1) allowing an additional unary constraint that forces a chosen variable to take
an integral value, and (2) identifying constraint languages which guarantee the existence
of integer solutions.

The reader should note that in the first approach we {\em do not} consider semilinear
relations defined over the integers. Instead, we consider ways of checking whether a given
problem instance has a solution where some variables are assigned integral values.
Some of the problems in the second approach can be seen as semilinear CSPs over
the integers, but our methods do not lend themselves to a systematic study of such CSPs.
See~\cite{Bodirsky:etal:icalp2015} for a recent approach to such a systematic study.

\subsection{The relation ${\mathbb Z}$}

The unary relation ${\mathbb Z}$ can be used to ensure that a variable is
given an integral value. By Lemma~\ref{fakeomin}, this relation is not semilinear over ${\mathbb Q}$,
so the constraint languages that we classify in the next theorem are formally not semilinear.

\begin{theorem}
\label{thm:main-int}
Let $\{R_+\} \subseteq \Gamma \subseteq SL_{\mathbb Q}[{\mathbb Q}]$ be a finite constraint language that satisfies (P$_0$) and (P$_\infty$).
The problem CSP$(\Gamma \cup \{{\mathbb Z}\})$ is in P if
\begin{itemize}
\item
$\Gamma$ is 0-valid; or
\item
$\langle \Gamma \rangle$ does not contain a bnu.
\end{itemize}
Otherwise, CSP$(\Gamma \cup \{{\mathbb Z}\})$ is NP-hard.
\end{theorem}

\begin{proof}
If $\Gamma$ is 0-valid, then $\Gamma \cup \{{\mathbb Z}\}$ is 0-valid, 
so every instance admits the solution $(0,0,\dots,0)$.

Otherwise, assume that $\langle \Gamma \rangle$ does not contain a bnu.
Let $\Phi$ be an arbitrary instance of CSP$(\Gamma \cup \{{\mathbb Z}\})$,
let $I \subseteq \vars(\Phi)$ be the set of variables that are constrained by ${\mathbb Z}$ in
$\Phi$,
and let $\Phi'$ be the instance of CSP$(\Gamma)$ obtained from $\Phi$
by removing all ${\mathbb Z}$-constraints.
By an argument on the set of all bounded, non-empty unary relations in $\langle \Gamma \rangle$
similar to that used in the proof of Theorem~\ref{thm:main}, it follows that $\langle \Gamma \rangle$
contains the relation $\{1\}$.
Let $S$ be the set of satisfying assignments to $\Phi'$
By running Algorithm~\ref{alg:affc},
we obtain a system of inequalities that defines the affine hull $A$ of the satisfying assignments $S$.

We now substitute each such inequality for an equality.
The resulting system of linear equalities still defines $A$.
Let $A' = \{ \pi_I(x) \mid x \in A \}$.
We can compute a system of linear equations for $A'$ in polynomial time by first
computing a parameter form for $A$, removing the coordinates not corresponding to $I$,
and then computing the equivalent system of linear equations.
This can be in polynomial time by being careful with the representation sizes
of the intermediary results (cf.~Schrijver~\cite[Section~3]{Schrijver:TLIP}).
We then solve the resulting system of linear equations for an integer solution in polynomial time (cf.~Schrijver~\cite[Corollary~5.3]{Schrijver:TLIP}).
If no such solution exists, then $\Phi$ is unsatisfiable. 
Otherwise, the integer points in $A'$ are given by 
$L = \{ c_0 + \sum_{i=1}^k \lambda_i c_i \mid \lambda_1, \dots, \lambda_k \in {\mathbb Z} \}$,
for some linearly independent vectors
$c_0, \dots, c_k \in {\mathbb Z}^{I}$, where $k = \dim(A')$.
The vectors $c_i$ can be found explicitly in polynomial time,
but since we are only interested in showing that there exists a satisfying assignment to $\Phi$,
it suffices that $L$ has the aforementioned form.

For $p \in A'$ and constant $\epsilon > 0$, define $B(p, \epsilon) = \{ x \in A' \mid \| p - x \| < \epsilon \}$.
Let $S' = \{ \pi_I(x) \mid x \in S \}$ be the projection of $S$ to the variables in $I$ and
note that $S' \in \langle \Gamma \rangle$.
By Lemma~\ref{lem:intersection}, 
$S'$ contains a linear set $R \subseteq {\mathbb Q}^{I}$ such that $\aff(R) = A'$.
Let $p \in R$ and $\epsilon > 0$ be such that $B(p,\epsilon) \subseteq R \subseteq S'$.
We claim that there exist distinct $q_1, q_2 \in L$ such that the line
through $q_1$ and $q_2$ intersects $B(p,\epsilon)$ in an open line segment.
Let $U = {\mathcal L}_{S', q_1, q_2}$.
Since $\langle \Gamma \rangle$ does not contain a bnu, it follows that 
$(M, \infty) \subseteq U$ for some $M < \infty$.
Therefore, $q' = q_1 + t (q_2-q_1) \in S'$ for a large enough integer $t$.
Hence, there exists a point $q \in S$ such that $\pi_I(q) = q'$,
so $\Phi$ is satisfiable.

To prove the claim, 
let $B = B(p,\epsilon)$ and let $q_1 \in L \setminus B$.
Consider the cone $C = \{ q_1 + t (x-q_1) \mid x \in B,  t \geq 0 \}$ and
note that $C$ contains $B' := \{ q_1 + \delta \epsilon^{-1} (x-q_1) \mid x \in B \} = B(q_1 + \delta \epsilon^{-1} (p-q_1), \delta)$.
For a large enough positive constant $\delta$, the set $B' \cap L$ is non-empty.
Let $q_2 \in B' \cap L \subseteq C$.
Then, the line through $q_1$ and $q_2$ intersects $B$ in an open line segment.

Finally, assume that $\langle \Gamma \rangle$ contains a bnu $U$.
We may assume that $U$ is not $0$-valid:
By Lemma~\ref{zerovalidlemma}, 
$\langle \Gamma \rangle$ contains a non-empty unary relation $T$
that is not $0$-valid.
Let $c \in {\mathbb Q}$ be a non-zero constant such that $U \cap c \cdot T \neq \emptyset$.
If $U \cap c \cdot T$ contains more than one element, then it is a bnu that is not $0$-valid.
Otherwise, $U \cap c \cdot T$ is a constant unary relation,
so $\langle \Gamma \rangle$ contains $\{1\}$.
In this case, for a large enough constant $c \in {\mathbb Q}$,
the relation $U + c \in \langle \Gamma \rangle$ is a bnu that is not $0$-valid.

Let $r_1, r_2 \in U$ be two distinct points and let $c \in {\mathbb Q}$ be a non-zero constant such that
$c \cdot r_1, c \cdot r_2 \in {\mathbb Z}$.
Then, $U' = c \cdot U \cap {\mathbb Z}$ is a bnu that excludes an interval and
$U' \cap (-1,1) = \emptyset$.
NP-hardness therefore follows from Lemma~\ref{universalhardness}.
\end{proof}

\subsection{The integer property}
In this section, we will determine those semilinear constraint languages containing $R_+$
for which knowing that there is a solution guarantees that there is an integer solution.
We make the following definition.

\begin{definition}
Let $\Gamma$ be a constraint language over ${\mathbb Q}$. We say that
$\Gamma$ has the {\em integer property} if
every instance of CSP$(\Gamma)$ has a solution if and only if it has an integer solution.
\end{definition}

The integer property can be used to infer tractability of certain semilinear constraint languages
over ${\mathbb Z}$.
In particular, if $\Gamma$ is a semilinear constraint language over ${\mathbb Q}$
that satisfies the integer property,
then CSP$(\Gamma)$ and CSP$(\Gamma|_{\mathbb Z})$ are equivalent.
To see that $\Gamma|_{\mathbb Z}$ is a semilinear constraint language over ${\mathbb Z}$,
take an arbitrary $R \in \Gamma$ and let $\phi$ be a quantifier-free 
definition of $R$ over $LI_{\mathbb Q}[{\mathbb Z}]$.
Then, $\phi$ is also a quantifier-free definition of $R|_{\mathbb Z}$ over $LI_{\mathbb Z}[{\mathbb Z}]$.

The following lemma shows that the integer property is preserved under pp-definitions.

\begin{lemma}\label{lem:ipclosed}
Let $\Gamma$ be a constraint language over ${\mathbb Q}$.
If $\Gamma$ has the integer property,
then so does $\langle \Gamma \rangle$.
\end{lemma}

\begin{proof}
Let $\Psi$ be an CSP-instance with relations $R_1, \dots, R_k$ from $\langle \Gamma \rangle$,
let $\phi_1, \dots, \phi_k$ be pp-definitions of $R_1, \dots, R_k$ in $\Gamma$, and let
$\Psi'$ be the CSP$(\Gamma)$-instance obtained from $\Psi$ by replacing each relation $R_i$ by
the quantifier-free part of $\phi_i$, and adding existential quantifiers for all auxiliary variables.
If $\Psi$ has a rational solution, then $\Psi'$ has a rational solution, so $\Psi'$ has an integer solution.
Note that the restriction of any solution of $\Psi'$ to $\vars(\Psi)$ is a solution to $\Psi$.
Therefore, the restriction of an integer solution of $\Psi'$ to $\vars(\Psi)$ is an integer solution
to $\Psi$, which proves the lemma.
\end{proof}

Let $\Gamma$ denote a semilinear constraint
language that contains $R_+$. Observe that if $\{1\} \in \langle \Gamma \rangle$, then
CSP$(\Gamma)$ cannot have the integer property since the
following CSP$(R_+ \cup \{1\})$-instance has the unique solution $x = \frac{1}{2}, y = 1$:
\[\exists x, y \: . \: x+x=y \wedge \{1\}y.\]

\begin{definition}
Let $\Gamma$ be a constraint language over ${\mathbb Q}$.
We say that $\Gamma$ is {\em scalable} if
the following holds: for each $R \in \Gamma$ and for each
$x=(x_1,\ldots,x_k) \in R$, 
there exists a positive constant $A$ such that
$(ax_1,\ldots,ax_k) \in R$, for all 
$a \geq A$.
\end{definition}

Clearly,
scalable constraint languages cannot contain any unary constant relation $\{c\}$ except when $c=0$.
Note that
if $\Gamma$ has endomorphisms $e(x)=a \cdot x$ for all rational $a > A > 0$, then
$\Gamma$ is indeed scalable. Inferring the existence of endomorphisms
from the scalability property is, in general, not straightforward or even possible.
The scalability property was originally defined slightly
differently~\cite{Jonsson:Loow:AI13} but it is easy to verify that the two
definitions coincide.

The following result completely characterises the semilinear constraint languages
that contain $R_+$ and have the integer property.

\begin{theorem} \label{unboundedinteger}
Let $\{R_+\} \subseteq \Gamma \subseteq SL_{\mathbb Q}[{\mathbb Q}]$ be a
constraint language that is not $0$-valid.
Then, the following are equivalent:

\begin{enumerate}
\item\label{ip:ip}
$\Gamma$ has the integer property.

\item\label{ip:unary}
every unary relation in $\langle \Gamma \rangle$
is either $\{0\}$ or unbounded.

\item\label{ip:infty}
$\Gamma$ does not satisfy (P$_\infty$).

\item\label{ip:scalable}
$\Gamma$ is scalable.
\end{enumerate}

\end{theorem}
\begin{proof}
(\ref{ip:ip}) $\Rightarrow$ (\ref{ip:unary}).
We show $\neg$(\ref{ip:unary}) $\Rightarrow$ $\neg$(\ref{ip:ip}).
Suppose that $T_1 \neq \{0\}$ is a bounded unary relation
in $\langle \Gamma \rangle$.
By Lemma~\ref{zerovalidlemma}, there is a non-empty unary relation $T_2$ in $\langle \Gamma \rangle$ that is not $0$-valid.
Therefore, for some $c \in {\mathbb Q}$, the unary relation 
$U = T_1 \cap c \cdot T_2$ in $\langle \Gamma \rangle$ is non-empty, bounded,
and not $0$-valid.
Let $k=1+\lceil \max(|\sup U|,|\inf U|) \rceil$.
Consider the CSP instance 
\[\exists x, y \: . \: U(x) \wedge k \cdot y=x,\]
and note that it has a solution: arbitrarily choose $x \in U$ and let $y=x/k$.
However, it cannot have any integer solution since $0 \not\in U$ and $k$ was chosen such that $k > |x|$.
Both $U$ and the equation $k \cdot y = x$ are pp-definable in $\Gamma$,
so the claim follows from Lemma~\ref{lem:ipclosed}.
\medskip

(\ref{ip:unary}) $\Rightarrow$ (\ref{ip:infty}).
We show $\neg$(\ref{ip:infty}) $\Rightarrow$ $\neg$(\ref{ip:unary}).
Assume that there exists a unary relation $U$ in $\langle \Gamma \rangle$ containing a positive
point and $(M,\infty) \cap U = \emptyset$, for some $M < \infty$.
If $U$ is bounded, then $\neg$(\ref{ip:unary}) follows immediately.
Otherwise, by Lemma~\ref{fakeomin},
there exists some $M' < \infty$ such that $(M',\infty) \cap U = \emptyset$ and $(-\infty,-M') \subseteq U$.
By Lemma~\ref{halfboundgenerate}(2), there exists a bounded unary relation
in $\langle \{R_+,U\} \rangle$ and, consequently, there exists such
a relation in $\langle \Gamma \rangle$.
\medskip

(\ref{ip:infty}) $\Rightarrow$ (\ref{ip:scalable}).
We show $\neg$(\ref{ip:scalable}) $\Rightarrow$ $\neg$(\ref{ip:infty}).
Arbitrarily choose an $n$-ary relation $R \in \Gamma$ such that
$R$ is not scalable. Arbitrarily choose a tuple $p=(p_1,\ldots,p_n) \in R$
that witnesses that $R$ is not scalable, i.e.,
the set $Y = \{ y \geq 1 \mid y \cdot p \in R \}$ is unbounded.
Consider the set $U=\{a \in {\mathbb Q} \mid a \cdot p \in R\}$
and note that $U$ is pp-definable in $\{R,R_+\}$ by Lemma~\ref{generateequations}:
\[U(x)\equiv\exists y_1,\dots,y_n \: . \: y_1=x \cdot p_1 \wedge \dots \wedge y_n=x \cdot p_n \wedge R(y_1,\dots,y_n).\]
Note that $1 \in U$ so $U$ contains a positive point.
Furthermore, since $Y$ is unbounded, it follows from Lemma~\ref{fakeomin}
that $(M,\infty) \subseteq Y$ for some $M < \infty$.
Hence, by definition of $U$, we have $(M,\infty) \cap U = \emptyset$,
so $\Gamma$ satisfies (P$_{\infty}$).
\medskip

(\ref{ip:scalable}) $\Rightarrow$ (\ref{ip:ip}).
This implication is not difficult to deduce from the proof of Lemma~6 in~\cite{Jonsson:Loow:AI13}.
We include
an argument here for completeness.
Assume that $\Gamma$ is scalable and let $\Phi$ be an arbitrary instance of CSP$(\Gamma)$
with a solution $x$.
Let $R_1, \dots, R_m$ be an enumeration of the atoms of $\Phi$ that contain a relation symbol from $\Gamma$.
Since $\Gamma$ is scalable, it follows that there exists a constant $A_i$
such that $a x$ satisfies $R_i(x_{i_1}, \dots, x_{i_k})$ for all $a \geq A_i$.
Let $A = \max \{ A_1, \dots, A_m \}$.
Then, $a x$ satisfies all atoms (including the equalities) of $\Phi$, for all $a \geq A$.
Therefore, if $a$ is chosen to be a large enough common multiple of the denominators in $x$,
then $a x$ is an integral solution to $\Phi$.
\end{proof}

As an immediate application of Theorem~\ref{unboundedinteger} we give a
complement to Theorem~\ref{thm:main-int} in
the case when $\Gamma$ violates (P$_\infty$).

\begin{corollary}
Let $\{R_+\} \subseteq \Gamma \subseteq SL_{\mathbb Q}[{\mathbb Q}]$ be a constraint language that violates (P$_\infty$).
The problem CSP$(\Gamma \cup \{{\mathbb Z}\})$ is in P if
\begin{itemize}
\item
$\Gamma$ is $0$-valid;
\item
$\langle \cone(\Gamma) \cup \{\{1\}\} \rangle$ does not contain a bnu; or
\item
$\cone(\Gamma)$ is essentially convex.
\end{itemize}
Otherwise, CSP$(\Gamma \cup {\mathbb Z})$ is NP-hard.
\end{corollary}

\begin{proof}
If $\Gamma$ is $0$-valid, then $\Gamma \cup \{{\mathbb Z}\}$ is $0$-valid,
and hence in P.
Otherwise, Theorem~\ref{unboundedinteger} implies that $\Gamma$ has the integer property.
Therefore, every instance of CSP$(\Gamma)$ has a solution if and only if it has an integer solution.
It follows that CSP$(\Gamma \cup \{{\mathbb Z}\})$ is polynomial-time equivalent to CSP$(\Gamma)$.
Since $\Gamma$ violates (P$_\infty$), the result follows from Theorem~\ref{nopropertymainresult}.
\end{proof}

\section{Discussion}
\label{sec:waysforward}

\subsection{Generalisations}\label{sec:generalisations}
A natural goal, following the proof of Theorem~\ref{superresult},
would be to determine the complexity of CSP$(\Gamma)$ 
for an arbitrary semilinear constraint language $\Gamma$, i.e.,
when $\Gamma$ does not necessarily contain $R_+$.
Below we indicate a few such attempts and the difficulties that accompany them. 

Consider Corollary~\ref{linearmainresult}. 
Our main result,
Theorem~\ref{superresult}, generalises this by
removing the assumption that $\{1\}$ is in $\Gamma$.
A natural question is then what happens if we instead remove the assumption
that the addition relation needs to be in $\Gamma$.
To this end, let $SL^1$ denote the set of semilinear constraint languages such
that $\{\{1\}\} \subseteq \Gamma$ and $\{R_+\} \not\subseteq \langle \Gamma \rangle$.
A straightforward modification of the construction in Section~6.3
of Jonsson and Lööw~\cite{Jonsson:Loow:AI13} gives the following:
 for every constraint language $\Gamma'$ over a finite domain, there
exists a $\Gamma \in SL^1$ such that 
CSP$(\Gamma')$ and CSP$(\Gamma)$ are polynomial-time equivalent problems.
Hence, a complete classification 
would give us a complete classification of finite-domain
CSPs, and such a classification is a major open
question within the CSP community~\cite{Feder:Vardi:stoc93,Feder:Vardi:siamjc98,Hell:Nesetril:08}. 
We also observe that for every
{\em temporal constraint language} (i.e., languages that are first-order
definable in $\{<\}$ over the rationals),
there
exists a $\Gamma \in SL^1$ such that 
CSP$(\Gamma')$ and CSP$(\Gamma)$ are polynomial-time equivalent problems.
This follows from the fact that every temporal constraint language $\Gamma'$
admits a polynomial-time reduction from
CSP$(\Gamma' \cup \{\{1\}\})$ to
CSP$(\Gamma')$: simply equate all variables appearing in $\{1\}$-constraints and note
that any solution can be translated into a solution such that this variable is assigned
the value 1. The complexity of temporal constraint languages is fully 
determined~\cite{Bodirsky:Kara:jacm2010} and
the polynomial-time solvable cases fall into nine different categories. 
The proof is complex and it is based on the universal-algebraic approach for studying CSPs.
We conclude that a complete classification of the languages in $SL^1$ will require advanced
techniques and will have to be conditioned on the classification of finite-domains CSPs.

A smaller first step towards removing $R_+$ from Corollary~\ref{linearmainresult}
would be to only slightly relax the addition relation.
Consider the \emph{affine} addition relation
$A_+=\{(a,b,c,d) \in {\mathbb Q}^4 \; | \; a-b+c=d\}$.
This relation can be viewed as a `relaxed' variant of $R_+$ since
$A_+$ can be pp-defined in $\{R_+\}$ but not the other way round.
Let $\Gamma$ be a constraint language such that 
$\{A_+,\{c\}\} \subseteq \Gamma \subseteq SL_{\mathbb Q}[\mathbb Q]$ for some
$c \in {\mathbb Q}$.
It is not hard to reduce the complexity classification for such constraint language
to that of Theorem~\ref{superresult}:

Given a relation $R \subseteq {\mathbb Q}^k$ and a rational number $c \in {\mathbb Q}$, let
$R+c$ denote the relation $\{(x_1+c,\dots,x_k+c) \mid (x_1,\dots,x_k) \in R\}$.
For instance, $A_+ + c = A_+$.
Similarly, we define $\Gamma+c=\{R+c \mid R \in \Gamma\}$ for constraint languages $\Gamma$.
Note that CSP$(\Gamma)$ and CSP$(\Gamma + c)$ are polynomial-time equivalent problems.

Arbitrarily choose
a constraint language $\{A_+,\{c\}\} \subseteq \Gamma \subseteq SL_{\mathbb Q}[\mathbb Q]$
and let $\Gamma'=\Gamma + (-c)$.
The problem CSP$(\Gamma')$ is polynomial-time equivalent with CSP$(\Gamma)$, $A_+ \in \Gamma'$
and $\{0\} \in \Gamma'$. 
The fact that $A_+ \in \Gamma'$
and $\{0\} \in \Gamma'$ implies that $R_+ \in \langle \Gamma' \rangle$ since
$R_+(x,y,z)$ can be pp-defined by
\[\exists w \: . \: \{0\}(w) \wedge A_+(x,w,y,z).\]
Consequently, CSP$(\Gamma')$ and CSP$(\Gamma' \cup \{R_+\})$ are polynomial-time
equivalent problems. We conclude that CSP$(\Gamma)$ is either in P or NP-complete
by Theorem~\ref{superresult}.

An interesting way forward would be to classify the complexity of CSP$(\Gamma)$
for all $\{A_+\} \subseteq \Gamma \subseteq SL_{\mathbb Q}[{\mathbb Q}]$.
Such a result would be a substantial generalisation of the results in Section~4
of Bodirsky et al.~\cite{Bodirsky:etal:jlc2012}.
Here, we see no obvious obstacles as in the case above for $\Gamma \in SL^1$.

\subsection{The metaproblem}

Theorem~\ref{superresult} shows that for every constraint language
$\{R_+\} \subseteq \Gamma \subseteq SL_{\mathbb Q}[{\mathbb Q}]$, the
problem
CSP$(\Gamma)$ is either in P or NP-complete.
This makes the following
computational problem (sometimes referred to as a {\em metaproblem} in the literature) relevant:

\medskip

\begin{center}
Given a constraint language $\{R_+\} \subseteq \Gamma \subseteq SL_{\mathbb Q}[{\mathbb Q}]$, is
CSP$(\Gamma)$ in P or NP-complete?
\end{center}

\medskip

We do not know the complexity of this problem and, in fact, 
it is not clear whether it is decidable or not.
Interesting methods for tackling similar questions have been identified by, for instance,
Bodirsky et al.~\cite{Bodirsky:etal:jsl2013} and
Dumortier et al.~\cite{Dumortier:etal:jcss99a,Dumortier:etal:jcss99b}.
Bodirsky et al. analyse the decidability of abstract properties of
constraint languages such as
whether certain relations are pp-definable or not.
Their results are based on a number of different techniques from model theory, universal algebra,
Ramsey theory, and topological dynamics.
Dumortier et al.~\cite{Dumortier:etal:jcss99a,Dumortier:etal:jcss99b} show
that it is decidable whether a given first-order formula using the binary functions
$*$ and $+$, and the binary relation $\leq$ over $\mathbb R$ with parameters from $\mathbb Q$ 
defines a semilinear relation. These results indicate that there are non-obvious properties
of semilinear relations that may be relevant for proving (un)decidability of the metaproblem.

\subsection*{Acknowledgements}
The authors thank Manuel Bodirsky for suggesting the relation $A_+$ as a relaxation of $R_+$ (cf.~Section~\ref{sec:generalisations}).

\end{document}